\documentclass[11pt,a4paper]{article}

\usepackage{url}
\usepackage{ifthen}                 
\usepackage{fancyhdr}               
\usepackage{longtable}              
\usepackage{afterpage}              
\usepackage[normal,footnotesize,bf]{caption2} 
\usepackage[round,sort]{natbib}       
\bibpunct{[}{]}{,}{n}{}{;}     
\usepackage[nottoc]{tocnatbibind}   
\makeatletter
\renewcommand\section{\@startsection{section}{1}{\z@}%
                       {-10\p@ \@plus \p@ \@minus 0\p@}%
                       {10\p@ \@plus 2\p@ \@minus 4\p@}%
                       {\normalfont\large\bfseries\boldmath
                        \rightskip=\z@ \@plus 8em\pretolerance=10000 }}
\renewcommand\subsection{\@startsection{subsection}{2}{\z@}%
                       {-5\p@ \@plus -4\p@ \@minus 2\p@}%
                       {8\p@ \@plus 4\p@ \@minus 4\p@}%
                       {\normalfont\normalsize\bfseries\boldmath
                        \rightskip=\z@ \@plus 8em\pretolerance=10000 }}
\renewcommand\subsubsection{\@startsection{subsubsection}{3}{\z@}%
                       {-18\p@ \@plus -4\p@ \@minus -4\p@}%
                       {-0.5em \@plus -0.22em \@minus -0.1em}%
                       {\normalfont\normalsize\bfseries\boldmath}}
\renewcommand\paragraph{\@startsection{paragraph}{4}{\z@}%
                       {-12\p@ \@plus -4\p@ \@minus -4\p@}%
                       {-0.5em \@plus -0.22em \@minus -0.1em}%
                       {\normalfont\normalsize\itshape}}
\renewcommand\subparagraph[1]{\typeout{LLNCS warning: You should not use
                  \string\subparagraph\space with this class}\vskip0.5cm
You should not use \verb|\subparagraph| with this class.\vskip0.5cm}

\DeclareMathSymbol{\Gamma}{\mathalpha}{letters}{"00}
\DeclareMathSymbol{\Delta}{\mathalpha}{letters}{"01}
\DeclareMathSymbol{\Theta}{\mathalpha}{letters}{"02}
\DeclareMathSymbol{\Lambda}{\mathalpha}{letters}{"03}
\DeclareMathSymbol{\Xi}{\mathalpha}{letters}{"04}
\DeclareMathSymbol{\Pi}{\mathalpha}{letters}{"05}
\DeclareMathSymbol{\Sigma}{\mathalpha}{letters}{"06}
\DeclareMathSymbol{\Upsilon}{\mathalpha}{letters}{"07}
\DeclareMathSymbol{\Phi}{\mathalpha}{letters}{"08}
\DeclareMathSymbol{\Psi}{\mathalpha}{letters}{"09}
\DeclareMathSymbol{\Omega}{\mathalpha}{letters}{"0A}

\def\indexname{Index}
\makeatother

\usepackage{amsmath,amsfonts,amssymb}
\usepackage{verbatim}
\usepackage{epsf}

\usepackage[bottom]{footmisc}
\usepackage{url}

\usepackage[pdftex]{graphicx}       
\DeclareGraphicsExtensions{.pdf}   
\usepackage[pdftex,                
bookmarks=true,
bookmarksnumbered=true,
hypertexnames=false,
breaklinks=true,
linkbordercolor={0 0 0.8},
pdfborder={0 0 0.6}]{hyperref}
\hypersetup{
pdfauthor   = {J. Ratsaby},
pdftitle    = {Information Width},
pdfsubject  = {Article},
pdfkeywords = {},
pdfcreator  = {LaTeX with hyperref package},
pdfproducer = {dvips + ps2pdf}
}

\newcommand{\real}{\mbox{\rm I\hspace{-0.2em}R}}

\newtheorem{theorem}{Theorem}

\newtheorem{remark}{Remark}
\newtheorem{lemma}{Lemma}

\newtheorem{definition}{Definition}

\newcommand{\Real}{ \real }

\def\prob{\mathbb{P}}

\def\sy{\mathsf{y}}

\def\eff{\eta}
\def\cost{\kappa}

\def\bY{\mathbb{Y}}
\def\bX{\mathbb{X}}
\def\bZ{\mathbb{Z}}
\newenvironment{proof}{\vskip 2mm\noindent {\it Proof}:}{\hfill $\square$ \vskip 2mm}

    \def\mF{{\mathcal     F}}    \def\mM{{\mathcal     M}}
\def\mV{{\mathcal V}}
\def\mQ{{\mathcal Q}}
 \def\mL{{\mathcal L}}
  \def\VC{\text{VC}}
  \def\VCs{\text{VC}}
    \def\L{\text{L}}
\def\tr{\text{tr}}     \def\mP{{\mathcal      P}}

\def\mH{\mathcal{H}}

\def\dist{\mbox{dist}}

\newcommand{{\integer}}{\mbox{\rm Z\hspace{-0.85em}Z}\,}
\newcommand{\be}{\begin{equation}}
\newcommand{\ee}{\end{equation}}

\newcommand{\bq}{\begin{eqnarray*}}
\newcommand{\eq}{\end{eqnarray*}}

\newcommand{\bqq}{\begin{eqnarray}}
\newcommand{\eqq}{\end{eqnarray}}

\def\squareforqed{\hbox{\rlap{$\sqcap$}$\sqcup$}}

\def\qed{\ifmmode\squareforqed\else{\unskip\nobreak\hfil
\penalty50\hskip1em\null\nobreak\hfil\squareforqed
\parfillskip=0pt\finalhyphendemerits=0\endgraf}\fi}

\def\blacksquareforqed{\hbox{$\blacksquare$}}
\def\bqed{\ifmmode\blacksquareforqed\else{\unskip\nobreak\hfil
\penalty50\hskip1em\null\nobreak\hfil\blacksquareforqed
\parfillskip=0pt\finalhyphendemerits=0\endgraf}\fi}

\def\triangleforqed{\hbox{$\triangle$}}
\def\tqed{\ifmmode\triangleforqed\else{\unskip\nobreak\hfil
\penalty50\hskip1em\null\nobreak\hfil\triangleforqed
\parfillskip=0pt\finalhyphendemerits=0\endgraf}\fi}

\usepackage{makeidx}  


\usepackage{lineno}

\begin{document}

 \thispagestyle{plain}

\title{Information width}
\author{ Joel Ratsaby\\
{\small Electrical and Electronics Engineering Department}\\
{\small Ariel University Center of Samaria}\\
{\small ISRAEL}\\
 {\small\tt ratsaby@ariel.ac.il}}


\maketitle

\begin{abstract}

Kolmogorov argued that the concept of information  exists also
in problems with no underlying stochastic  model (as Shannon's information representation)
for instance,  the information contained in an algorithm or in the genome.
He introduced a combinatorial notion of entropy 
and information $I(x:\sy)$  conveyed
by a binary string  $x$
about the unknown value of a variable $\sy$.
The current paper poses the following questions: what is the relationship between
the information conveyed by  $x$ about $\sy$ to 
the  description complexity of  $x$ ?
is there a notion of cost of information   ?  
are there limits on how efficient  $x$ conveys information ?
 To answer these questions Kolmogorov's definition is extended
 and a new concept 
termed {\em information width} which is similar to $n$-widths in approximation theory is introduced.
Information of any input  source, e.g., sample-based, general side-information or a hybrid of both
can be evaluated by a single common formula.
An application to the space of binary functions is considered.

\end{abstract}


\section{Introduction}
\label{intro}

Kolmogorov \cite{Kolmogorov65} sought for a measure of information of  `finite objects'. 
He considered three approaches, the so-called {combinatorial},
{probabilistic} and {algorithmic}.
The probabilistic approach corresponds to the well-established definition of the Shannon entropy
 which applies to  stochastic settings where an `object' is represented by a random variable.
In this setting,  the entropy of an object and the  information conveyed by
one object about another  are well defined.
Here it is necessary to view an object (or a finite binary string)
as a realization of a stochastic process. While this has often been used, for instance, to measure
the information of  English texts \cite{CoverKing1978,Kontoyiannis97} by assuming some finite-order Markov process, it is not obvious
that such modeling of finite objects provides a natural and a universal representation of information
as  Kolmogorov states in \cite{Kolmogorov65}:
{\em What real meaning is there, for example, in asking
how much information is contained in (the book) "War and Peace" ? Is it reasonable
to ... postulate some probability distribution
 for this set ?
 Or, on the other hand, must we assume that the individual scenes in this book form a random sequence with stochastic relations that damp out quite rapidly over a distance of several pages ?
 }
These questions  led
 Kolmogorov to  introduce an alternate  non-probabilistic and algorithmic notion  of the  information contained in a finite
binary string. He defined  it as the length of the minimal-size program that can compute the string.
This has been
later developed into 
the  so-called 
{\em Kolmogorov Complexity} field   \citep{LiVitanyi97}.

In the combinatorial approach, 
Kolmogorov investigated another non stochastic measure of information 
for an object $y$. 
Here $y$ is taken to be any
element in 
a finite space $\bY$ of objects.
In \cite{Kolmogorov65} he defines the `entropy' of  
 $\bY$  as $H(\bY) = \log |\bY|$ where $|\bY|$ denotes
 the cardinality of $\bY$ and
all logarithms henceforth are
taken with respect to $2$.
As he writes, if the value  of 
$\bY$ is known to be $\bY=\{y\}$ then this much entropy 
is `eliminated' by providing $\log |\bY|$ bits of `information'.

Let $R=\bX \times \bY$ be a general finite domain
and consider a set 
\be
\label{A}
A\subseteq R
\ee
 that consists of all `allowed' values of pairs $(x,y)\in R$.
The entropy of $\bY$ is defined as 
\[
H(\bY)=\log|\Pi_\bY(A)|
\]
 where
\(
\Pi_\bY(A) \equiv  \{y\in \bY: (x,y)\in A \text{ for some } x\in \bX\}
\)
denotes the projection of  $A$ on $\bY$.
Consider the restriction of $A$ on $\bY$ based on $x$ which is defined as
\be
\label{Ax}
Y_x =\{y\in \bY: (x,y) \in A\}, \; x\in\Pi_\bX(A)
\ee
  then
 the conditional combinatorial  entropy of $\bY$ given $x$ is defined as
 \be
  \label{Hxy}
 H(\bY|x) = \log |Y_x|.
 \ee
Kolmogorov  defines  the  information conveyed by  $x$ about  $\bY$ by the quantity 
\be
\label{KIxy}
I(x:\bY) = H(\bY) - H(\bY|x).
\ee 
Alternatively, we
may view $I(x:\bY)$ as
the information
that a set $Y_x$ conveys about another set $\bY$ satisfying $Y_x \subseteq \bY$.
In this case we  let the domain be $R=\Pi_\bY(A)\times \Pi_\bY(A)$,
 $A_x\subseteq R$ 
is the set of permissible pairs
 $A_x = \{(y, y'): y\in \Pi_\bY(A), y'\in Y_x\}$
and  the information is defined as
\be
\label{btts}
I(Y_x: \bY) = \log|\Pi_\bY(A)|^2 - \log(|Y_x| |\Pi_\bY(A)|).
\ee
We will refer to this representation as Kolmogorov's  information between  sets.
Clearly, $I(Y_x:\bY) = I(x:\bY)$.

In many applications,  knowing an input
 $x$ 
 only conveys partial information
  about an unknown value $y\in\bY$.
For instance, in  problems 
which involve  the analysis of algorithms on discrete classes of structures, 
such as sets of binary vectors or functions on a finite domain,
 an algorithmic search is made for some optimal element in this set
based only on partial information.
One such paradigm is the area of  statistical pattern recognition \cite{Vapnik1998,AB99} 
 where an unknown target, i.e., a pattern classifier, is seeked 
 based on the information contained in a finite sample and some side-information.
This information is implicit in the particular set of classifiers that form  the possible hypotheses.

For example, let  $n$ be a positive integer
and   consider the domain $[n]=\{1, \ldots, n\}$. Let $F=\{0, 1\}^{[n]}$ be the set
   of all binary functions $f:[n]\rightarrow\{0, 1\}$.
  The power set $\mP(F)$ represents  the family of all  sets
   $G\subseteq F$. Repeating this, we have  $\mP(\mP(F))$  as the collection of all properties of sets $G$, i.e.,
   a {\em property} is a set whose elements are  subsets $G$ of $F$.
  We  denote by $\mM$ a property of a set $G$ and write $G\models\mM$.
      Suppose that we seek to know  an unknown target function $t\in F$.
Any  partial  information about   $t$ which may be expressed by $t\in G\models \mM$
can effectively reduce
the search space. 
It has been a long-standing problem  to try to quantify 
 the value of general side-information for learning (see \cite{RatsabyMaiorov98} and references therein).

We assert  that Kolmogorov's combinatorial framework may serve as a basis.
We let 
      $x$ index possible
  properties $\mM$  of  subsets $G\subseteq F$
  and the  object $y$ represent the unknown  target $t$ 
  which may be any element of $F$.
  Side information is then represented by knowing certain properties
  of sets that contain the  target.
The input $x$
conveys that $t$ is in some subset $G$  that has a certain property $\mM_x$.   

In principle, Kolmogorov's quantity
 $I(x:\bY)$  should  serve as the value of information in $x$ about the unknown value of $y$.
However,  its current form (\ref{KIxy})   is not general enough
 since it requires that the target $y$ be restricted to
a fixed set $Y_x$ on knowledge of  $x$.
To see this,
suppose $t$ is in a set that satisfies property $\mM_x$.
Consider the collection $\{G_z\}_{z\in Z_x}$ of all subsets $G_z\subseteq F$ that have this property.
Clearly, $t\in \bigcup_{z\in Z_x}G_z$ hence
 we may first consider  $Y_x = \bigcup_{z\in Z_x}G_z$ but
 some useful information  implicit in this collection is ignored as we now show:
 consider two properties $\mM_0$ and $\mM_1$ 
 with corresponding
  index sets $Z_{x_0}$ and $Z_{x_1}$ such that
  $\bigcup_{z\in Z_{x_0}}G_z = \bigcup_{z\in Z_{x_1}}G_z\equiv F'\subseteq F$.
  Suppose that most
 of the sets $G_z$, $z\in Z_{x_0}$ are small while
 the sets
  $G_z$, $z\in Z_{x_1}$ are large. Clearly, property $\mM_0$
  is more informative than $\mM_1$ since
  starting with  knowledge
    that  $t$ is in  a set that satisfies it should take (on average) 
    less additional
  information (once the particular set $G$  becomes known) in order to completely specify $t$.
    If, as above, we let $A_{x_0}= \bigcup_{z\in Z_{x_0}}G_z$
 and $A_{x_1}= \bigcup_{z\in Z_{x_1}}G_z$ then
we have  $I(x_0:\bY) = I(x_1:\bY)$ which wrongly implies
  that both properties are equally informative.
Knowing  $\mM_0$ provides implicit information associated with 
the collection of possible sets $G_z$,  $z\in Z_{x_0}$.
  This implicit structural information  cannot be represented
  in  (\ref{KIxy}). 

\section{Overview}

In \cite{Ratsaby2006a} we began to consider an extension of Kolmogorov's combinatorial information
that can be  to applied for more  general settings.
The current paper further builds upon this 
and continues to
 explore the `objectification' of information, viewing it 
 as a `static' relationship between  sets of objects
in contrast to the standard Shannon representation.
As it is based on basic set theoretic principles no assumption is necessary
concerning  the underlying space of objects other than its finiteness.
It 
is thus  more fundamental than the standard  probability-based representation
 used in information theory. 
It is also more general than
  Bayesian approaches, for instance in statistical pattern recognition,
    which assume that  a target $y$ 
is randomly drawn from $\bY$ according to a prior probability distribution.

    The main two contributions of the paper are,
    first, the introduction of a set-theoretic framework 
    of information and its efficiency, and
     secondly
  the application of  this framework 
to classes of binary functions.
Specifically, in Section \ref{opt} 
we define 
a quantity called the information width (Definition \ref{optd}) which measures  
the information conveyed about  an unknown target $y$
by a maximally-informative input  of a fixed description
 complexity $l$ (this notion is defined in Definition \ref{lx1}).
The first result,  Theorem \ref{optt}, computes this width and it is consequently   used
 as a  reference point
 for a comparison of the information value of different inputs.
This is done via 
the measures of cost and efficiency  of information   defined in Definitions \ref{eff} and \ref{eff1}.
The width serves as a universal reference against which any type
of information input
may be computed and compared, for instance, the information 
of a finite data sample used in a problem of learning
can be compared to 
other kinds of side-information.

 In Section \ref{sec4}
we apply the framework  to the space of binary functions on a finite domain.
We consider  information which is conveyed via properties of  classes of binary functions,
specifically, those that relate to the complexity of learning such functions.
The properties  are stated in terms of  combinatorial quantities such as
 the  Vapnik-Chervonenkis (VC) dimension.
Our interest in investigating the information conveyed by such properties stems from
the  large body of work  on  learning binary function
 classes (see for instance \cite{Mansour94,145265,28427}).
This is part of 
the area called  statistical learning theory that  deals with computing the
complexity 
of learning over   hypotheses classes, e.g., neural networks,  using 
 various   algorithms  each
 with its particular type of  side information (which is sometimes
 referred to as the `inductive bias', see \cite{Mitchell1997}).

For instance  it is known \cite{AB99} that learning an unknown target function
in a class of VC-dimension (defined in Definition \ref{VC}) which is no greater than  $d$ 
requires a training sample of size  linear in $d$.
The knowledge  that
  the target is in a class which has this property  conveys
the useful information 
  that there exist an algorithm
which  learns any target in this class
  to within   arbitrarily low  error
  based on a finite samples
   (in general this is impossible
if  the VC-dimension is infinite).
 But  how valuable
is this information,  can it be quantified and
 compared to other types of side information ?
 
The theory  developed here
 answers this and treats the problem of computing the
  value of  information
 in a uniform manner for any source.
The generality of this approach
stands on  its set-theoretic basis. Here  information, or entropy, is defined 
in terms of the number of bits that it takes to index objects in general sets.
This is conveniently  applicable to  settings such as those in learning theory  where
the underlying structures consist of classes of objects or
 inference models (hypotheses) such as   binary functions (classifiers),
  decision trees, Boolean formulae, neural networks.
Another area (unrelated to learning  theory) which can be applicable 
is  computational biology.
Here a sequence of the nucleotides or
amino acids  make up DNA, RNA or protein
molecules and one is interested in the amount of functional information
needed to specify sequences with internal
order or structure.
Functional information is not a property of
any one molecule, but of the ensemble of all
possible sequences, ranked by activity.
As an example,  \cite{SzJack2003} considers a pile of DNA, RNA or protein
molecules of all possible sequences, sorted
by activity with the most active at the top.
More information
is required to specify molecules that
carry out difficult tasks, such as high-affinity
binding or the rapid catalysis of chemical
reactions with high energy barriers, than is
needed to specify weak binders or slow
catalysts. But, as stated in  \cite{SzJack2003}, precisely how much more
functional information is required to
specify a given increase in activity is
unknown. 
Our model may be applicable here if we let all sequences of a fixed  
level $\alpha$ of activity 
be in the same class that satisfies  a property $\mM_\alpha$.
Its  description complexity (introduced later) may represent
 the amount of functional information.

In Section \ref{sec4} we consider an application of this model
and state  Theorems \ref{p1} -- \ref{cor2}
which  estimate the information value,
 the  cost and the description
complexity 
associated with
 different properties. 
  This allows to compute  their information efficiency and  
  compare their relative values.
 For instance, Theorem \ref{cor2} considers a hybrid property which
 consists of two sources of information, a sample of size $m$
  and side-information about the VC-dimension $d$ of the class that contains the target.
 By computing the efficiency with respect to
  $m$ and $d$ we determine  how  effective is each of these sources.
   In  Section \ref{sec4a} we compare the information efficiency 
of several additional properties of this kind.

\section{Combinatorial formulation of information}
 \label{sec3}

 In this section we  extend the information measure (\ref{KIxy})
 to one that  applies to a more general setting (as discussed in Section \ref{intro})
 where  knowledge of  $x$  may still leave some vagueness about the possible  value of $y$.
 As in \cite{Kolmogorov65} we seek a non-stochastic representation 
    of the information conveyed by  $x$ about $y$. 
  
  Henceforth let $\bY$ be a general finite domain, let $\bZ = \mP(\bY)$ and $\bX = \mP(\bZ)$ where as before
  for a set $E$ we denote by  $\mP(E)$ its power set.
  Here $\bZ$ represents the set of  indices $z$ of all possible sets $Y_z\subseteq\bY$ and 
  $\bX$ is the set of  indices $x$ of all possible collections, i.e., subsets
   $Z_x\subseteq \bZ$ of indices of  sets $Y_z$, $z\in \bZ$. 
   We say that a set $Y_z\subseteq \bY$ has a property $\mM_x$ 
   if $z\in Z_x$ where $Z_x$ is the uniquely corresponding collection of $\mM_x$.
   
  The previous representation based on (\ref{A})
  is subsumed by  this representation since  instead of  $\bX$ we have $\bZ$
  and  the sets $Y_x$, $x\in \Pi_\bX(A)$  defined in (\ref{Ax}) 
  are now represented by the sets $Y_z$,  $z\in Z_x$.
  Since  $\bX$ indexes  all possible properties of sets in $\bY$
  then for any $A$ as in (\ref{A}) there exists an $x\in \bX$ in the new representation such that
  $\Pi_\bX(A)$ in Kolmogorov's representation is equivalent to $Z_x$ in the new representation.
  Therefore what was previously represented by the  sets $\{Y_x: x\in \Pi_\bX(A)\}$ 
  is  now the collection of sets
  $\{Y_z: z\in Z_x\}$.
  In this new representation   a given input $x$
  can point to multiple subsets $Y_z$ of $\bY$, $z\in Z_x$, and hence apply for
  the more general settings discussed in the previous section.

We  will view  the information conveyed by $x$ about an unknown
object $y$ through two perspectives. The first is  held by the side that
provides the information  and the second by the side which acquires it.
From the side of the provider, we denote by the subset $\sy\subset\bY$ 
a  set of target values $y$, for instance,  solutions to a problem any one of which the provider
may wish to inform the acquirer. 
In general the provider provides  partial
information about $y$ via
an object $x\in\bX$ which is used as a  means of representing this information
and   as input 
to the acquirer.

From the acquirer's side, initially (before seeing $x$) 
the set of possible targets is the whole target-domain $\bY$ since
he  does not `know' the subset $\sy$.
After seeing the input  $x$ he then has a collection of sets $Y_z$, $z\in Z_x$,
one of which is ensured (by the provider) to intersect the subset $\sy$.
In this case we say that $x$ is  informative about  $\sy$
(Definition \ref{informative}).
 Kolmogorov's representation fits  the  acquirer's perspective
where   the unknown subset $\sy$ is just the whole  target domain $\bY$ (known by default)
 and therefore $x$ is the only variable in the information formula  of (\ref{KIxy}).
 In all subsequent   definitions that involve $\sy$ 
we may  switch between the two perspectives
simply by replacing  $\sy$ with $\bY$.

\begin{definition}
\label{informative}
Let  $\sy\subseteq\bY$ be fixed.
An input object  $x\in\bX$ is called {\em informative} for $\sy$, denoted $x\vdash \sy$,
if there exists a $z\in Z_x$ with a corresponding  set $Y_z$
 such that
 $\sy\bigcap Y_z\neq \emptyset$.
\end{definition}
The following 
is our definition of the combinatorial  value of information.
\begin{definition}
\label{I}
Let $\sy\subseteq \bY$ and consider any $x\in\bX$  such that $x\vdash \sy$.
Define by
\[
I(z: \sy) = \log (|\sy\bigcup Y_z|^2) - \log(|\sy|| Y_z|).
\]
%
Then the information conveyed by $x$ about the unknown value $y\in\sy$ is defined as
\[
I(x:\sy) \equiv  \frac{1}{|Z_x|} \sum_{z\in Z_x} I(z:\sy).
\]
\end{definition}
\begin{remark}
For a fixed $x$,  
the information value $I(x:\sy)$ is in general dependent on $\sy$ 
since different   $\sy$ for which $x$ is informative ($x\vdash\sy$)
may have different values of $I(z: \sy)$, $z\in Z_x$.
The information value is non-negative real measured in bits.
\end{remark}
Henceforth, it will be convenient to assume that
 $I(x:\sy)=0$ whenever
   $x$ is not informative for $\sy$.
\begin{remark}
We will refer to $I(x:\sy)$ as the provider's (or provided) information 
about the unknown target $y$ given $x$.
The acquirer's (or acquired)
 information is  defined based on
 the special case where $\sy = \bY$. Here  we have 
\(
|\sy\bigcup Y_z| = |\bY\bigcup Y_z| = |\bY|
\)
and the information value becomes
\bqq
I(x:\bY)
&=& \frac{1}{|Z_x|}\sum_{z\in Z_x}\left[2\log|\bY| - \log|\bY| - \log|Y_z|  \right]\nonumber\\
\label{newIxy}
&=& \log|\bY| - \frac{1}{|Z_x|}\sum_{z\in Z_x} \log|Y_z|.
\eqq
\end{remark}
Definition \ref{I} is consistent with  (\ref{KIxy}) in that
 the representation of uncertainty is done
as in \cite{Kolmogorov65}
in a set-theoretic approach since
all expressions in (\ref{newIxy}) involve set-quantities such as cardinalities and restrictions of sets.
The expression of (\ref{KIxy}) is a special case of (\ref{newIxy})
 with $Z_x$  being a singleton set
 and $\sy = \bY$.
In defining $I(z:\sy)$
we have implicitly extended Kolmogorov's
 information 
$I(Y_x:\bY)$ between   sets $Y_x$ and  $\bY$
that satisfy $Y_x \subseteq \bY$ (see (\ref{btts}))
into the more general
definition where one of the two sets is not necessarily contained in the other and
 neither one equals
 the whole space $\bY$, i.e., $I(z:\sy) \equiv I(Y_z:\sy)$
is the  information between the sets $Y_z$ and $\sy$
where  $Y_z$ is not necessarily contained in $\sy$.
Here we take  the underlying two-dimensional domain 
 as $R=(\sy\bigcup Y_z) \times (\sy\bigcup Y_z)$
and the set  of permissible pairs as
\[
A_{z,\sy} = \{(y, y'): y\in \sy, y'\in  Y_z\} \subset R.
\]
We may view the relationship between the provider and acquirer as a transformation 
between sets
\[
\sy\rightarrow \{Y_z: z\in Z_x\} {\rightarrow} \bY
\]
where the provider, knowing the set $\sy$, chooses some $x$ with which
he represents the unknown value of $y$ and for him, the amount
of information remaining about $y$ as conveyed by $x$
is $I(x:\sy)$ bits.
The acquirer, starting from knowing only  $\bY$,
 uses $x$, or equivalently the corresponding collection  of sets $\{Y_z: z\in Z_x\}$, 
as an intermediate `medium' to acquire $I(x:\bY)$ bits of information
about the unknown value of $y$.
Note that, in general, 
the provider's information may be smaller, equal
or larger than the acquired information. 
For instance, fix an $x$, then directly from Definition \ref{I}
 for any $z\in Z_x$ we can compare 
$I(z:\sy)$ versus $I(z:\bY)$ and see that if $|\sy\bigcup Y_z|$ is closer to $|\bY|$ (or $|\sy|$)
than to $|\sy|$ (or $|\bY|$) then
$I(z:\sy) > I(z:\bY)$ (or $I(z:\bY) > I(z:\sy)$) respectively.
Thus taking the average over all $z\in Z_x$ it is possible in general to have $I(x:\sy)$ 
smaller, larger or equal to  $I(x:\bY)$.

In this paper 
we will primarily use  the acquirer's perspective 
and will thus  refer to the sum in (\ref{newIxy}) as the conditional combinatorial entropy which is defined next.
\begin{definition}
Let
\be
\label{condEnty}
H(\bY|x) \equiv  \frac{1}{|Z_x|}\sum_{z\in Z_x} \log|Y_z|
\ee
be the conditional entropy of $\bY$ given $x$.
\end{definition}
It will be convenient to express the conditional entropy as
\[
H(\bY|x) = \sum_{k\geq 2}  \omega_{x}(k) \log k
\]
with
\be
\label{omegax}
\omega_{x}(k)=\frac{|\{z\in Z_x: |Y_z|=k\}|}{|Z_x|}.
\ee
We will refer to this quantity $\omega_{x}(k)$
as the conditional {\em density} function of $k$. 
The factor of $\log k$ comes from $\log |Y_z|$ which from (\ref{Hxy})
is the combinatorial
 conditional-entropy $H(\bY|z)$.



\section{Description complexity}

We have so far defined the notion of  information $I(x:\sy)$ about the unknown value $y$
 conveyed by  $x$.
 Let us now define the description complexity of $x$.

%
%
\begin{definition}
\label{lx}
The {\em description complexity}  of $x$, denoted $\ell(x)$,
is defined as
\be
\label{lx1}
\ell(x) \equiv \log \frac{|\bZ|}{|Z_x|}.
\ee
\end{definition}
\begin{remark}
The description complexity $\ell(x)$ is a positive real number  measured in bits.
It takes a fractional value  if  the cardinality of $Z_x$ is greater than half
that of $\bZ$.
\end{remark}

Definition \ref{lx} is motivated from the following:
 from Section \ref{sec3}, the input  $x$ conveys
a certain property common to every  set $Y_z$,  $z\in Z_x\subseteq \bZ$,
 such that the unknown value $y$
is an element of at least one such set $Y_z$.
Without the knowledge of $x$
these indices $z$  are only known to be elements of $\bZ$ in which
case it takes 
 $\log|\bZ|$ bits to describe any $z$, or equivalently, any $Y_z$.
 If  $x$ is given
  then the length of the binary  string  that describes a $z$ in $Z_x$
 is only $\log|Z_x|$. 
 The set $Z_x$ can therefore be described by a string of length  $\log|\bZ| - \log|Z_x|$
 which is precisely the right side of (\ref{lx1}). Alternatively, $\ell(x)$
 is the information $I(x: \bZ)$ gained about the unknown value $z$ 
 given $x$ (since $x$ points to a single set $Z_x$  then this  information
 follows directly from Kolmogorov's formula (\ref{KIxy})).
  
As $|Z_x|$ decreases 
 there are fewer possible  sets $Y_z$ that satisfy the property described by $x$
and the description complexity $\ell(x)$ increases.
In this case, $x$
conveys a more `special'  property of the possible sets $Y_z$ and  the `price'
of describing such a property increases.
The following is a useful result.
\begin{lemma}
Denote by $Z^c_x = \bZ\setminus Z_x$ the complement of the set $Z_x$ and let $x^c$ denote the input
 corresponding to $Z^c_x$. Then  
\[
\ell(x^c) = -\log (1- 2^{-\ell(x)}).
\]
\end{lemma}
\begin{proof}
 Denote by $p=|Z_x|/|\bZ|$. Then by definition of the description complexity we have 
 \(
 \ell(x^c) = -\log (1-p).
 \)
 Clearly, 
 \[
 2^{-\log \frac{1}{1-p}} =1 - 2^{-\log\frac{1}{p}}
 \]
 from which the result follows. 
\end{proof} 
\begin{remark}
\label{p1mp}
Since
clearly
 the proportion of elements $z\in\bZ$ which are in $Z_x$ plus the 
 proportion of those in $Z^c_x$ is fixed and equals $1$ then
 $2^{-\ell(x)} + 2^{-\ell(x^c)} = 1$.
If the description complexity $\ell(x)$ and $\ell(x^c)$ change
(for instance with respect to an increase in $|\bZ|$)
 then they change in opposite directions.
However, the cardinalities of the corresponding sets $Z_x$ and $Z^c_{x}$ may both increase (\ref{lx1})
  for instance if
  $|\bZ|$ grows at a rate 
 faster than the rate of change of either $\ell(x)$ or $\ell(x^c)$.
\end{remark}

A  question to raise at this point is whether the following trivial relationship  between $\ell(x)$ and the entropy $H(\bY|x)$ holds, 
\be
\label{test}
\ell(x) + H(\bY|x) \stackrel{?}{=} H(\bY).
\ee
This is equivalent to asking if
\be
\label{test1}
\ell(x)  \stackrel{?}{=} I(x:\bY)
\ee
or in words, does the price of describing an input $x$ equals the information gained by knowing it  ?

As we show next, the answer depends on
 certain characteristics of the set $Z_x$. When
 (\ref{KIxy}) does not apply but (\ref{newIxy}) does,  then in general,
  the relation does  not hold.

\section{Scenario examples}

In all the following scenarios we take the acquirer's perspective, i.e., 
with no input given the unknown $y$ is only known to be in $\bY$.
As the first scenario, we start with the simplest uniform setting which is defined as follows:\\
\mbox{}\\\noindent {\em Scenario S1}: As in (\ref{Ax})-(\ref{KIxy}),  an input $x$ amounts to a single set $Y_x$.
The set $Z_x$ is a singleton $\{Y_x\}$ so  $|Z_x|= 1$ and instead of  
$\bZ$ we have $\Pi_\bX(A)$.
We impose the following conditions:
for all $x, x' \in \bX$, 
 $Y_x \bigcap Y_{x'} = \emptyset$,  $|Y_x| = |Y_{x'}|$.
With  $\bY = \bigcup_{x\in\bX}Y_x$ then
it follows that for any $x$, $|Y_x| = \frac{|\bY|}{|\bX|}$.
From (\ref{lx}) it follows that the description complexity of any $x$ is 
\[
\ell(x) = \log \frac{|\bX|}{1} = \log |\bX|
\]
and  the entropy 
\[
H(\bY|x) = \log|Y_x| =\log \left(\frac{|\bY|}{|\bX|}\right).
\]
We therefore have
\[
\ell(x) + H(\bY|x) = \log |\bX| + \log |\bY| - \log |\bX| = \log|\bY|.
\]
Since the right side equals $H(\bY)$ then (\ref{test}) holds. 
Next consider another scenario:\\
\mbox{}\\\noindent {\em Scenario S2}: An input $x$  gives a single set $Y_x$
but now for any two distinct $x, x'$, we only force the condition
that $|Y_x| = |Y_{x'}|$, i.e., the intersection $Y_x\bigcap Y_{x'}$ may be non-empty.
The  description complexity $\ell(x)$ is the same as in the previous scenario
and  for any $x, x'\in\bX$
the entropy is the same $H(\bY|x) = H(\bY|x')$
with a value of $\log \left(\frac{\alpha|\bY|}{|\bX|}\right)$, for some $\alpha \geq 1$.
So 
\[
\ell(x) + H(\bY|x) = \log |\bX| + \log\left(\frac{\alpha|\bY|}{|\bX|}\right) =
 \log\left(\alpha|\bY|\right) \geq \log|\bY|.
\]
Hence the left side of (\ref{test}) is greater than or equal to the right side.
By (\ref{test1}), this means that the `price', i.e., the description complexity
 per bit of information may be  larger than $1$. 

Let us introduce at  this point  the following combinatorial quantity:
\begin{definition}
\label{eff}
The {\em cost} of  information $I(x:\sy)$, denoted $\cost_\sy(x)$, is defined as
\[
\cost_\sy(x) = \frac{\ell(x)}{I(x:\sy)}
\]
and  represents the number of description bits of $x$  per bit
 of information  about the unknown value of $y$ as conveyed by $x$.
\end{definition}
Thus, letting $\sy=\bY$ and considering  the two previous  scenarios where
Kolmogorov's definition (\ref{KIxy}) applies then
 the cost of information equals $1$ or at least $1$, respectively.
As the next scenario, let us consider the following:\\
\mbox{}\\\noindent {\em Scenario S3}: We follow the setting of Definition \ref{I} where
an input $x$  
means that the unknown value of $y$ is contained in at least one set $Y_z$, $z\in Z_x$
hence  $|Z_x|\geq 1$.
Suppose that 
$\frac{|\bY|}{|\bX|} \equiv a$ for some integer $a\geq 1$ and
assume that
 for all  $x\in\bX$, $H(\bY|x) = \log(a)$.
(The sets $Y_z$, $z\in Z_x$  may still differ in size and  overlap).
Thus 
we have
\be
\label{inequ5}
\ell(x) + H(\bY|x) = \log \frac{|\bZ|}{|Z_x|} + \log\frac{|\bY|}{|\bX|}.
\ee
Suppose that 
$\frac{|\bZ|}{|\bX|} \equiv b$ for some integer $b\geq 1$,
$Z_{x} \bigcap Z_{x'}=\emptyset$ and $|Z_x| = |Z_{x'}|$ for any $x, x'\in \bX$.
Since
$\bZ = \bigcup_{x\in\bX} Z_x$ then
$|Z_x|=b$  for all $x\in\bX$.
The right side of (\ref{inequ5})
equals $\log|\bY|$ and (\ref{test}) is satisfied.
If for some $x, x'$ we have  $|Z_x|<b$ and $|Z_{x'}|>b$ (with entropies both still at $\log(a)$) then 
the left side of (\ref{test}) is greater than or less than $H(\bY)$, respectively. Hence
it is possible in this scenario for the cost $\kappa_\sy(x)$ to be  greater or less than $1$.
 To understand why for some inputs $x$
the cost may be strictly smaller than $1$ observe that
under the current scenario the actual set $Y_z$ which contains 
the unknown $y$ remains unknown even after
producing the description $x$.
Thus in this case the left  side of (\ref{test}) represents 
the total description complexity of the unknown value of $y$
 (on average over all possible sets $Y_z$)
given that the {\em only} fact known about $Y_z$
is that its index $z$ is an element of $Z_x$.
In contrast, scenario {\em S1} has  the total description complexity of the unknown $y$
on the left  side of (\ref{test}) which also
 includes the description of the specific  $Y_x$ that contains $y$ (hence  it may be longer).
Scenario {\em S3} is an example, 
as mentioned in Section \ref{intro}, 
of  knowing a property which still leaves the acquirer
with several  sets that   contain the unknown $y$.

In Section \ref{sec4} we will consider several specific properties of this kind.
Let us now continue and introduce  additional concepts as part of the framework.

\section{Information width and efficiency}
\label{opt}

With the definitions of Section \ref{sec3} in place
we now have  a quantitative measure  of the information (and cost)
 conveyed by
 an input $x$ about an unknown value $y$. 
 This  $y$ is contained in some
 set that satisfies a certain property and the set itself
 may remain unknown.
 In   subsequent sections we  consider several  examples of 
 inputs $x$ for which  these measures are computed and compared.
 Amongst the  different ways of conveying   information about an unknown value $y$ it is natural
to ask at this point
if there exists
a notion of maximal information.
This is formalized next by the following definition  which resembles
 $n$-widths used in functional approximation theory \cite{Pinkus85}.
\begin{definition}
\label{optd}
Let
\be
\label{IAB}
I_p^*(l)\equiv\max_{
         \substack{x\in\bX \\ \ell(x)=l}
         }
 \;\min_{
 	 \substack{ \sy \subseteq\bY\\ x\vdash \sy  } 
 	 } I(x:\sy)
\ee
be the {\em $l^{th}$-information-width}.
\end{definition}
\begin{remark}
The above definition is stated from the provider's point of view. He 
is free to choose a fixed `medium', i.e., a structure $Z_x$ of sets (but  limited
in its description complexity to $l$)
in order to provide information at some later time about any
 set  $\sy\subseteq \bY$ of objects to the acquirer. For that he
 considers all possible inputs $x$ of description complexity  $l$
 and measures the information it will provide for 
 the hardest target-subset $\sy$.
 We refer to the above as the provider's information width.

If we set $\sy=\bY$ then we obtain the acquirer's width of information, denoted as 
\[
I_a^*(l) \equiv I^*(l)
\]
 which takes a simpler form of
 \[
 I^*(l) = 
 \max_{
         \substack{x\in\bX \\ \ell(x)=l}
         }
 	  I(x:\bY).
  \]
  \end{remark}

The next result computes the value of  $I^*(l)$.

\begin{theorem}
\label{optt}
Denote by $\mathbb{N}$ the positive integers.
Let $1\leq l\leq \log|\bZ|$ and define
\[
r(l) \equiv \min\left\{ a\in\mathbb{N}: \sum_{i=1}^a {|\bY| \choose i} \geq |\bZ|2^{-l}\right\}.
\]
Then we have
\be
\label{istar}
I^*(l) = \log |\bY| - 
\frac{2^l}{|\bZ|} \left(
\sum_{k=2}^{r(l)-1}{|\bY|\choose k}\log k + \left(|\bZ|2^{-l}  - 
\sum_{i=1}^{r(l)-1}{|\bY|\choose i}\right)\log r(l)\right).
\ee
\end{theorem}
\begin{proof}
Consider
a particular input  $x^*\in\bX$ with a description complexity $\ell(x^*)=l$ and with
a corresponding $Z_{x^*}$ that contains the indices $z$ 
of as many distinct non-empty sets $Y_z$ of the lowest possible cardinality.
By (\ref{lx1}) it follows that $Z_{x^*}$ satisfies $|Z_{x^*}| = |\bZ|2^{-l}$ and
contains all $z$ such that 
  $1\leq |Y_z|\leq r(l)-1$
  in addition to $|\bZ|2^{-l}  - \sum_{i=1}^{r(l)-1}{|\bY|\choose i}$ elements $z$
  for which $|Y_z|=r(l)$.
We therefore have $I(x^*:\bY)$ as equal to the right side of (\ref{istar}). Any other $x$ with $\ell(x)=l$
  must
have $H(\bY|x)\geq H(\bY|x^*)$ since it is formed by replacing one of the sets $Y_z$ above
with a larger set $Y_{z'}$. Hence for such $x$, $I(x:\bY) \leq I(x^*:\bY)$ and
therefore $I^*(l) = I(x^*:\bY)$.
\end{proof}
 The notion of width is more general than that defined above. 
  For instance in functional approximation theory
 the so called $n$-widths are used to measure the approximation error of some rich general class of functions, 
 e.g., Sobolev class, by the closest element of a manifold of simpler function classes. For instance,
the  Kolmogorov width $K_n(F)$ of a class $F$ of functions  (see \cite{Pinkus85})  
is defined as $K_n(F) = \inf_{F_n\subset F}\sup_{f\in F}\inf_{g\in F_n}\|f-g\|$
where $F_n$ varies over all linear subspaces of $F$ of dimensionality $n$.
Thus from this more general set-perspective it is  perhaps not surprising  that such a basic quantity of width
has also an information theoretic interpretation as we have shown in (\ref{IAB}).
  The work of
  \cite{DBLP:journals/dam/MaiorovR98}  considers the  {\em VC-width} of a finite-dimensional set $F$
defined as
\[
\rho^{VC}_n(F)\equiv \inf_{H^n}\sup_{f\in F} \dist(f,H^n)
\] 
where $F\subset \Real^m$ is a target set, $H^n$ runs over the class $\mH^n$ of all sets
 $H^n\subset\Real^m$ of VC-dimension $\VC(H^n)=n$ (see Definition \ref{VC}) and 
 $\dist(f,H^n) \equiv \inf_{h\in H^n}\dist(f,h)$ where $\dist()$ denotes the distance 
  between an element $f\in F$ 
 and  $h\in H^n$
 based on the $l^m_q$-norm, $1\leq q\leq \infty$.
We can make the following analogy with the information width of (\ref{IAB}):
$f$ corresponds to $\sy$,
   $F$ to $\bY$,
$h$ to $z$,
 $n$ corresponds to $l$,
$H^n$ to $x$ (or equivalently to $Z_x$),
the condition $\VC(H^n)=n$ corresponds to 
the condition of having a description complexity  $\ell(x)=l$,
the class $\mH^n$ corresponds to the set $\{x\in \bX: \ell(x) = l\}$,
$\dist(f,h)$ corresponds to $I(z:\sy)$,
$\dist(f, H^n)=\inf_{h\in H^n}\dist(f,h)$ corresponds to  $I(x:\sy) = (1/|Z_x|)\sum_{z\in Z_x}I(z:\sy)$,
$\sup_{f\in F}$ corresponds to $\min_{\sy\subset \bY}$,
and
$\inf_{H^n: \VCs(H^n)=n}$ corresponds to $\max_{x: \ell(x)=l}$.

The notion of information efficiency to be introduced below is based on 
the acquirer's information width $I^*(l)$.
\begin{definition}
\label{eff1}
Denote by
\[
\cost^*(x) \equiv \frac{\ell(x)}{I^*(\ell(x))}
\]
the per-bit cost 
 of  maximal  information   conveyed about an unknown target  $y$ in $\bY$
considering  all possible inputs of the same description complexity as $x$.
  Consider an input $x\in\bX$
 informative for  $\bY$.
   Then the {\em efficiency} of $x$ for $\bY$ is defined by
\[
\eff_\bY(x)  \equiv \frac{\cost^*(x)}{\cost_\bY(x)}
\]
where  the cost  is defined in Definition \ref{eff}.
\end{definition}
\begin{remark}
By definition of $\cost^*(x)$ and $\cost_\sy(x)$
it follows that
\be
\label{effAlt}
\eff_\bY(x) = \frac{I(x:\bY)}{I^*(\ell(x))}.
\ee
\end{remark}
While we will not use it here,  the provider's efficiency can be defined in a similar way.

Let us consider  an example where the above definitions may be applied.
Let   $n$ be a positive integer and denote
 by  $[n]=\{1, \ldots, n\}$.   
Let the target space be $\bY=\{0,1\}^{[n]}$ which consists of all binary functions $g:[n]\rightarrow\{0,1\}$.
Let $\bZ=\mP(\bY)$ be the set of  indices  $z$ of all possible classes  $Y_z\subseteq\bY$ of binary functions $g$ on $[n]$
(as before 
for any set $E$ we denote by $\mP(E)$ its  power set).
Let $\bX=\mP(\bZ)$ consist of all possible (property) sets $Z_x\subseteq\bZ$.
Thus here every possible class of binary functions 
  on $[n]$
and
every possible property of a class 
is represented.
Figure \ref{Istara}(a) shows  $I^*(l)$ and Figure \ref{Istara}(b) displays the cost $\cost^*(l)$ for this example 
as $n=5,6,7$.
From these graphs we see that the width $I^*(l)$ grows at a sub-linear rate  with respect to $l$ 
since the cost strictly increases.

In the next section, 
we apply the theory introduced in the previous sections
to the space of binary functions.


\section{ Binary function classes}
\label{sec4}

Let
$F=\{0, 1\}^{[n]}$
 and write  ${\cal P}(F)$ for the power set which consists
  of all subsets $G\subseteq F$.
  Let  $G\models\mM$ represent the statement ``$G$ satisfies property $\mM$''.
In order to apply  the above framework
we let
 $y$ represent an unknown target $t\in F$
and  $x$  a description object, e.g., a binary string, that describes the
possible properties $\mM$ of sets $G\subseteq F$
which may  contain  $t$. 
Denote by  $x_\mM$  the object that describes property $\mM$.
Our aim is to compute the value of information $I(x_{\mM}:F)$, the description complexity $\ell(x_{\mM})$,  the cost
$\cost_F(x_{\mM})$ and efficiency $\eff_F(x)$ for 
various inputs $x_\mM$.

Note that the set $Z_x$ used  in the previous sections is now a collection of classes $G$, i.e.,
elements of $\mP(F)$, which satisfy a property $\mM$.
We will sometimes refer to this collection   by $\mM$ and write 
$|\mM|$ for its cardinality (which is analogous to $|Z_x|$ in the notation of the preceding sections).

Before we proceed, let us recall a few basic definitions from set theory.
For any fixed subset $E\subseteq[n]$ of cardinality $d$ and any $f\in F$
 denote by $f_{|E}\in\{0, 1\}^d$ the restriction of $f$ on $E$.
For a set  $G\subseteq F$ of functions, the
 set 
 \[
 \tr_G(E)=\{f_{|E}: f\in G\}
 \]
  is called the {\em trace} of $G$ on $E$. The trace is a basic and useful measure of the combinatorial richness of 
 a binary function class
and is related  to its density  (see Chapter
17 in \cite{Bollobas86}).
It has also been shown to relate to various fundamental results in different fields, e.g.,
   statistical 
learning theory \cite{Vapnik1998},
combinatorial geometry \cite{PachAgrawal1995},
graph theory \cite{hw-esrq-87,AnthonyBrightwellCooper95}
and the theory of empirical processes \cite{p-csp-84}.
It is a member of a more general class of  properties that are expressed in terms of certain  allowed or forbidden
restrictions \cite{AnsteeFlemingFurediSali05}.
In this paper we focus on
 properties    based  on the trace of a class 
 which are expressed in terms of a positive integer parameter $d$ in the following general form:
\[
d = \max\{|E|: E\subseteq[n], \text{ condition on } \tr_G(E) \text{ holds}\}.
\]

The first definition taking such form is the so-called Vapnik-Chervonenkis dimension \cite{VapnikChervonenkis71}.
\begin{definition}
\label{VC}
The Vapnik-Chervonenkis dimension of a set $G\subseteq F$, denoted $\VC(G)$, is defined as
\[
\VC(G) \equiv \max\{|E|: E\subseteq[n], |\tr_G(E)|=2^{|E|}\}.
\]
\end{definition}
The next definition considers the other  extreme for the size of the trace.
\begin{definition}
Let $\L(G)$ be defined as
\[
\L(G) \equiv \max\{|E|: E\subseteq[n], |\tr_G(E)|=1\}.
\]
\end{definition}
For any $G\subseteq F$ define the following three properties:
 \begin{eqnarray*}
\mL_d &\equiv& \mbox{`$\L(G)\geq d$'}\\
\mV_d &\equiv & \mbox{`$\VC(G) < d$'}\\
\mV^c_d&\equiv&\mbox{`$\VC(G)\geq d$'}.
\end{eqnarray*}
We now  apply the framework to these
and  other related
properties
(for clarity, we defer some of the proofs to Section \ref{proofs}).
Henceforth,
  for two sequences $a_n$, $b_n$,  we write $a_n\approx b_n$ to denote that $\lim_{n\rightarrow\infty}\frac{a_n}{b_n}=1$
  and $a_n \ll b_n$ denotes $\lim_{n\rightarrow\infty}\frac{a_n}{b_n}=0$.
  Denote the standard normal probability distribution and cumulative distribution by
$\phi(x) = (1/\sqrt{2\pi})\exp(-x^2/2)$ and $\Phi(x) = \int_{-\infty}^x \phi(z)dz$, respectively.
The main results  are stated as Theorems \ref{p1} through \ref{cor2}.
\begin{theorem}
\label{p1}
Let $t$ be an unknown element of $F$. 
 Then the value of information  in knowing
that $t\in G$ where  $G\models\mL_d$,  is 
\bq
I(x_{\mL_d}:F) &=& \log|F| -  \sum_{k\geq 2}\omega_{x_{\mL_d}}(k)\log k\\
 &\approx& n - 
 \frac{
 \Phi\left(-a\right)\log\left(\frac{2^n}{1+2^{d}}\right) + 
2^{-(n-d)/2}\phi(a) + O(2^{-(n-d)})}
{1 - \left(\frac{2^d}{1+2^d}\right)^{2^n}}
\eq
where
\[
a =  2(1+2^{d})2^{-(n+d)/2} - 2^{(n-d)/2}
\]
 and the description complexity of $x_{\mL_d}$ is
\[
\ell(x_{\mL_d}) \approx 2^n\left(\frac{2^d}{1+2^d}\right) - d- c\log n
\]
for some $1\leq c \leq d$, as $n$ increases.
\end{theorem}
\begin{remark}
For large $n$, we have the following estimates:
\[
I(x_{\mL_d}:F) \simeq n - \log\left(\frac{2^n}{1+2^d}\right) \simeq  d
\]
and
\[
\ell(x_{\mL_d}) \simeq 2^n - d.
\]
The cost is estimated by 
\[
\cost_F(x_{\mL_d}) \simeq \frac{2^n-d}{d}.
\]
 \end{remark}

The next result is for  property $\mV^c_d$. 
\begin{theorem}
\label{p2}
Let $t$ be an unknown element of $F$.
Denote by
\(
a = (2^n-2^{d+1})2^{-n/2}.
\)
Then the value of information  in knowing
that $t\in G$,  $G\models\mV^c_d$, is
\bq
I(x_{\mV^c_d}:F) &=& \log|F| -  \sum_{k\geq 2}\omega_{x_{\mV^c_d}}(k)\log k \\
 &\approx& n -  \frac{(n-1)\left(2^{n}\Phi\left(a \right)    + 2^{n/2}\phi\left(a \right) 
  \left(1+\frac{a^2}{(n-1)2^n}\right)  \right)}
{2^n\Phi(a) + 2^{n/2}\phi(a)}
\eq
with increasing $n$.
Assume that  $d=d_n > \log n$ then the description complexity of $x_{\mV^c_d}$ satisfies
\[
\ell(x_{\mV^c_d}) \approx d (2^d + 1) + \log(d) - \log(2^n\Phi(a) + 2^{n/2}\phi(a)) -\log n + 1.
\]
\end{theorem}
\begin{remark}
For large $n$,
 the information  value is approximately
 \[
 I(x_{\mV^c_d}:F)\simeq 1
 \]
 and
 \[
 \ell(x_{\mV^c_d}) \simeq d 2^d -n -\log\left(\frac{n}{d}\right)
 \]
 thus
  \[
 \cost_F(x_{\mV^c_d}) \simeq \ell(x_{\mV^c_d}).
 \]
  We note that the description length increases with respect to $d$ implying that
 the proportion of classes with a VC-dimension  larger than $d$ 
 decreases with  $d$.  With respect to $n$ it behaves oppositely.

\end{remark}

The property of having an (upper) bounded VC-dimension (or trace) has been widely studied 
in  numerous fields (see the earlier discussion).
For instance in statistical learning theory \cite{Vapnik1998,BoucheronBousquetLugosi2004a}
the important property of convergence of the empirical averages to the means occurs
 uniformly over all elements of an infinite
 class provided that it satisfies this property.
It is thus interesting to study the property
\(
{\mV_d}
\)
defined above even for a finite class of binary functions.
\begin{theorem}
\label{cor1}
Let $t$ be an unknown element of $F$.
The  value of information  in knowing that $t\in G$, $G\models{\mV}_d$ is
\bq
I(x_{{\mV}_d}:F)  &\approx&  1 - o(2^{-n/2})
\eq
with  $n$ and $d=d_n$ increasing such that  $n < d_n2^{d_n}$.
The  description complexity of $x_{{\mV}_d}$ is
\[
\ell(x_{ {\mV}_d }) = 
-\log \left(
1- 2^{
-\ell(	
x_{{\mV}^c_d} 
)	
} 	
  \right)
\]
where $\ell(x_{{\mV}^c_d})$ is as in Theorem \ref{p2}.
\end{theorem}
\begin{remark}
Both the description complexity and the cost of information are approximated as
  \[
 \cost_F(x_{{\mV}_d}) \simeq  \ell(x_{\mV_d}) \simeq -\log(1-2^{-(d2^d -n -\log(\frac{n}{d}))}) .
 \]
 Relating to Remark \ref{p1mp},  while $\ell(x_{\mV_d})$ increases with respect to $n$
   and hence
 the proportion of classes with the  property $\mV_d$ decreases as $2^{-\ell(x_{\mV_d})}$,  the actual number of 
 binary function classes that have this property
   (i.e., the cardinality of the corresponding set $Z_x$)
 {\em increases} with $n$ since 
 \[
 |Z_{x_{\mV_d}}| = |\bZ|\,2^{-\ell({x_{\mV_d}})} = 2^{2^n} \left(1-2^{-(d2^d -n -\log(\frac{n}{d}))}\right).
 \]
    The number of classes that have the complement property $\mV^c_d$ also clearly increases
    since $\ell(x_{\mV^c_d})$ decreases with $n$.
 We note that the description length decreases with respect to $d$ implying that
 the proportion of classes with a VC-dimension no larger than $d$ 
 increases with  $d$.
   \end{remark}

As another related case, consider 
an input $x$
which in addition to
conveying  that $t\in G$ with $\VC(G)<d$  also provides
a labeled sample $S_m = \{(\xi_i, \zeta_i)\}_{i=1}^m$,
$\xi_i\in [n]$, $\zeta_i = t(\xi_i)$, $1\leq i\leq m$.
This means that
for all $f\in G$, $f(\xi_i) = \zeta_i$,
 $1\leq i\leq m$. 
We express this by stating that $G$ satisfies the property
\[
\mV_d(S_m) \equiv  \mbox{`$\VC(G) < d$, $G_{|\xi}=\zeta$'}
\]
where $G_{|\xi}$ denotes the set of restrictions  $\{f_{|\xi}: f\in G\}$, $f_{|\xi} = [f(\xi_1), \ldots, f(\xi_m)]$ and
$\zeta = [\zeta_1, \ldots, \zeta_m]$.
The following result states the value of  information and cost for property $\mV_d(S_m)$.
\begin{theorem}
\label{cor2}

Let $t$ be an unknown element of $F$ and
  $S_m = \{(\xi_i,t(\xi_i))\}_{i=1}^m$  a sample.
Then the  value of information in knowing that $t\in G$ where $G\models{\mV}_d(S_m)$ is
\bq
I(x_{{\mV}_d(S_m)}:F)  &\approx&  m - o(2^{-(n-m)/2})
\eq
with  $n$ and $d=d_n$ increasing such that  $n < d_n2^{d_n}$.
The  description complexity of $x_{{\mV}_d(S_m)}$ is
\[
\ell(x_{ {\mV}_d(S_m) }) \approx  2^n(1+\log(1-p)) +   \frac{n-m}{d2^{d(1+2^d)}}\left(\Phi(a)2^{n-m} + \phi(a)2^{(n-m)/2}  \right)+
 \left(1-p\right)^{2^n} 
\]
where $p=2^{-m}/(2^{-m}+1)$, $a=(2^np-2^d)/\sigma$,  $\sigma=\sqrt{2^{n}p(1-p)}$.
\end{theorem}
\begin{remark}
The description complexity 
is estimated by
\[
\ell(x_{ {\mV}_d(S_m) }) \simeq 2^n\left(1 + \frac{n-m}{d2^{d(1+2^d)+m}} +\log(1-p)\right)
\]
and the cost of information is
 \[
 \cost_F(x_{{\mV}_d(S_m)}) \simeq \frac{\ell(x_{ {\mV}_d(S_m) })}{m}.
 \]
 \end{remark}
 \begin{remark}
The dependence of the description complexity on $d$
 disappears rapidly with increasing $d$,  the effect of $m$ 
remains minor which effectively makes $\ell(x_{ {\mV}_d(S_m) })$
almost take  the maximal possible value of $2^n$.
Thus the proportion of classes which satisfy property ${\mV}_d(S_m)$
is very small.
   \end{remark}

\subsection{Balanced properties}

Theorems  \ref{p2} and \ref{cor1}  pertain to property $\mV_d^c$ and its complement $\mV_d$.
It is interesting that in both cases the information value is approximately equal to $1$.
If we denote by  $P^*_{n,k}$  a uniform probability distribution over the space of classes $G\subset F$ conditioned
on $|G|=k$ (this will be defined later in a more precise context   in (\ref{uniformclass}))
then, as is shown  later,  $P^*_{n,k}(\mV_d)$
and $P^*_{n,k}(\mV^c_d)$  vary approximately linearly with respect to $k$. 
Thus in both cases the conditional density (\ref{omegax}) 
is dominated by the value of $k=2^{n-1}$ and hence both have approximately the same conditional entropies (\ref{condEnty})
and information values.
Let us 
define the following:
\begin{definition}
A property $\mM$ is called {\em balanced} if 
\[
I(x_{\mM}:F) = I(x_{\mM^c}:F).
\]
\end{definition}
We may characterize some sufficient conditions for $\mM$ to be balanced.
First, as in the case of  property $\mV_d$ and more generally for any property $\mM$
a  sufficient condition for this to hold 
 is to have a  density  (and that of its complement $\mM^c$) dominated by 
some cardinality value $k^*$. 
Representing  $\omega_{x_{\mM}}(k)$
by a posterior probability function $P_n(k|\mM)$, for instance as in (\ref{omega2}) for $\mM=\mL_d$,
makes
the conditional entropies $H(F|x_\mM)$ and $H(F|x_{\mM^c})$ be  approximately the same. 
A  stricter sufficient condition is to have
\[
\omega_{x_\mM}(k) = \omega_{x_{\mM^c}}(k)
\]
for every $k$.
This implies  the condition  that
\[
P_n(k|\mM) = P_n(k|\mM^c)
\]
which using Bayes rule gives 
\[
\frac{P(\mM^c|k)}{P(\mM|k)} = \frac{P(\mM^c)}{P(\mM)}, \quad \text{ for all }2\leq k \leq 2^n.
\]
In words, this condition says that
%
%
the bias of favoring a class $G$ as satisfying 
property $\mM$  versus  $\mM^c$ (i.e., the ratio of their probabilities)
should be constant with respect to the cardinality $k$ of $G$. 
Any such   property
is therefore characterized by  certain features
of a class $G$ that are invariant to its size, i.e., 
  if the size of $G$ is provided in advance
 then no information is gained about whether $G$ satisfies $\mM$ or its complement $\mM^c$.

In contrast, property $\mL_d$ is an example of a very unbalanced property.
It is an example of a general property  whose posterior
 function decreases fast with respect to $k$ as we now consider:

\mbox{}\\\noindent {\em Example}:
Let $\mQ$ be a property with a distribution $P^*_{n,k}(\mQ) = c\alpha^k$, $0 < \alpha < 1$, $c>0$.
In a similar way as Theorem \ref{p1} is proved we obtain that the information value of this property tends to
\[
I(x_{\mQ}:F) \approx n- \frac{\Phi(-a)\log \left(2^n\alpha/(1+\alpha)\right)  +  \phi(a)/\sqrt{\alpha 2^n} + O(1/(\alpha2^n)) }
{1-(1+\alpha)^{-2^{n}}}
\]
with increasing $n$ where $a=(2-2^np)/\sqrt{2^np(1-p)}$ and $p=\alpha/(1+\alpha)$. This is approximated as
\[
I(x_{\mQ}:F) \simeq n - \left(n + \log\left( \frac{\alpha}{1+\alpha}\right)\right) = \log \left(1+ \frac{1}{\alpha}\right).
\]
For instance, suppose $P^*_{n,k}(\mQ)$ is an exponential probability function  then taking  $\alpha = 1/e$ gives an information value of 
\[
I(x_\mQ:F) \simeq \log(1+e) \simeq  1.89.
\]
For the complement $\mQ^c$, if we approximate $P^*_{n,k}(\mQ^c)= 1- c\alpha^k \simeq 1$ and
the conditional entropy (\ref{condEnty}) as
\[
\frac{\sum_{k\geq 2}P^*_{n,k}(\mQ^c)P_n(k)\log k}{\sum_{j\geq 1}P^*_{n,j}(\mQ^c)P_n(j)} \simeq \sum_{k\geq 2}P_n(k)\log k \approx \log (2^{n-1}) = n-1,
\]
where $P_n(k)$ is the binomial probability distribution  with parameter $2^n$ and $1/2$, then
 the information value is approximated by
\[
I(x_{\mQ^c}:F) \simeq n - (n-1) = 1. 
\]
By taking $\alpha$ to be even smaller we obtain a property $\mQ$ which has a very
different information value compared to $\mQ^c$.

\section{Comparison}
\label{sec4a}
We now compare the information values and the efficiencies for the various inputs $x$ considered in the previous section.
In this comparison we also include the following simple property defined next:
let $G\in{\cal P}(\{0,1\}^n)$ be any class of functions and denote by the {\em identity property}
$\mM(G)$ of $G$  the `property which is satisfied only by $G$'.
We immediately have
\be
\label{caseI}
I(x_{\mM(G)}:F) = n-\log |G|
\ee
and
\[
\ell(x_{\mM(G)})=2^n - \log(1) = 2^n 
\]
since the cardinality $|\mM(G)| = 1$. The cost in this case is
\[
\cost_F(x_{\mM(G)}) = \frac{2^n}{n-\log|G|}.
\]
Note that  $x$ conveys that
 $t$ is in a specific class $G$  hence the entropy and information values
 are according to   Kolmogorov's  definitions (\ref{Hxy}) and (\ref{KIxy}).
The efficiency in this case is simple to compute using (\ref{effAlt}): we have
 $I^*(\ell(x)) = I^*(2^n)$ and the  sums in (\ref{istar}) vanish since $r(2^n)=1$ thus $I^*(2^n)=n$
and $\eta_F(x)= (n-\log|G|)/n$.

Let us first compare the information value and the efficiency of three subcases of this identity property
with the following three different class cardinalities:  $|G|= \sqrt{n}$, $|G|=n$ and $|G|=2^{n-\sqrt{n}}$.
Figure \ref{fig2a}  displays the information value and Figure \ref{fig2b} shows the efficiency for these subcases.
As seen
 the information value increases as the cardinality of $G$ decreases which follows from (\ref{caseI}).
The efficiency $\eff_F(x)$ for these three subcases may be obtained exactly and equals (according to the same
order as above) $1 - (\log n)/(2n)$, $1 - (\log n)/n$ and $1/\sqrt{n}$. Thus 
a property with a single element $G$ may have an efficiency
which  increases or decreases depending on the rate of growth of the cardinality of $G$ with respect to $n$.

Let us  compare  the  efficiency for inputs $x_{\mL_d}$, $x_{{\mV}^c_d}$, 
$x_{{\mV}_d}$ and $x_{{\mV}_d(S_m)}$. As an example, suppose that  the VC-dimension parameter $d$ grows as $d(n) = \sqrt{n}$.
As can be seen from Figure \ref{fig3a},  property ${\mV}_d$  is the most efficient of the three  staying above the  $80\%$ level.
Letting the sample size increase at the rate of $m(n)=n^a$ 
then from Figure \ref{fig3b} the efficiency of   $\mV_d(S_m)$ increases with  respect to $a$
 but remains smaller
than the efficiency of property $\mV_d$.
Letting
the VC-dimension increase as $d(n)=n^b$ then Figure \ref{Fig4} displays
 the  efficiency of  $\mV_d(S_m)$ as a function of   $b$
for several values of $a=0.1, 0.2, \ldots, 0.4$ where  $n$ is fixed at $10$.
 As seen, the efficiency increases approximately linearly with $a$
and non-linearly with respect to $b$ with a saturation  at approximately $b=0.2$.

\section{Conclusions}

The information width introduced here is  a fundamental 
 concept based on which  a combinatorial
interpretation  of information 
is defined and used as the basis for the concept of  efficiency of  information.
We defined the width  from two perspectives, that of the provider and the acquirer of information
and used it as a reference point according to which the efficiency of 
any   input information can be evaluated.
As an application we considered the space of binary function classes on a finite domain
and computed the efficiency of information conveyed by various class properties.
The main point that arises from these results is that
  side-information of different types can be quantified, computed and compared in this common
 framework
 which is  more general
 than the standard framework used in 
 the theory of information transmission.

As further work, it will  be interesting to compute the efficiency of 
information  in other   applications,
for instance, pertaining to properties of 
classes of Boolean functions $f:\{0,1\}^n \rightarrow\{0, 1\}$
(for which there are many applications, see for instance \cite{hammerYves}).
It will  be interesting to examine standard search algorithms, for instance, those used in machine learning
over a finite search space (or hypothesis space) and compute their information efficiency, i.e.,
accounting for all side information available for an algorithm (including data) and computing for it
the acquired information value and efficiency.

In our treatment of this subject
we  did not touch the issue  of how the information is used.
For instance, a learning algorithm uses side-information
and  training data  to  learn  a pattern classifier which has
minimal prediction (generalization)  error.
A search algorithm in the area of information-retrieval uses an input query
to return an answer set that overlaps as many of the relevant objects
and at the same time has as few  non-relevant objects as possible.
In each such application the information acquirer, e.g., an algorithm,
 has an associated performance criterion, e.g.,
 prediction error, percentage recall  or precision,  according to  which it is evaluated.
What is 
 the relationship
between information  and performance, does performance depend on efficiency
 or only on the amount of provided information ?
 what are the consequences  of using  input information of low efficiency ?
For the current work, we leave these questions as open.
The remaining parts of the paper consist of the technical work used to obtain the previous results.


\section{Technical results}

In this section we provide the proofs of Theorems \ref{p1} to \ref{cor2}.
Our approach is to estimate the number of sets $G\subseteq F$ 
that satisfy  a property $\mM$.
Using the techniques from \cite{YcartRatsaby06_1} we employ a probabilistic method by which a random class is generated
and the probability that it satisfies $\mM$ is computed. 
As we use the uniform probability distribution on elements of the power set $\mP(F)$ then
 probabilities yield  cardinalities of the corresponding sets. 
The computation of $\omega_{x}(k)$ and hence of (\ref{newIxy})
follows directly.
It is worth noting 
 that, as in \cite{Kolmogorov63},  
the notion of probability
is only used  here for simplifying some of the counting arguments
and thus, unlike Shannon's information, it plays no role 
in the actual definition of  information.

Before proceeding with the proofs, in the next section we describe the probability model for generating a random class.

\subsection{Random class generation}

In this subsection we describe  the underlying probabilistic processes 
with which a random class is generated. 
We use the so-called {\em binomial model} to generate a random class of binary functions (this
has been extensively used  in the area of random graphs  \cite{Jansonetal00}).
In this model,  the random class
  $\mF$ is constructed through
   $2^n$ independent  coin tossings, one for each function in  $F$,
   with a probability of success (i.e., selecting a function into $\mF$) equal to $p$.
The probability distribution  $P_{n,p}$ is  formally defined on $\mP(F)$  as follows:
given parameters $n$ and $0\leq p\leq 1$, for any $G\in{\cal P}(F)$,
\[
P_{n,p}(\mF=G) =
p^{|G|}(1-p)^{2^n-|G|}.
\]
 In our application, we choose $p=1/2$ and denote the probability distribution as
 \[
 P_n\equiv P_{n,\frac{1}{2}}.
 \] 
It is clear that
 for any element $G\in \mP(F)$,  the probability that the random class $\mF$ equals $G$ is 
 \be
 \label{unif1}
  \alpha_n \equiv P_n(\mF = G) =  \left(\frac{1}{2}\right)^{2^n}
\ee
and the probability of $\mF$ having a cardinality $k$ is
\be
\label{pfk}
P_n(|\mF|=k) = {2^n\choose k}\alpha_n, \;\; 1\leq k\leq 2^n.
\ee
The following fact easily follows from the definition of the conditional probability:
for any set  $B\subseteq{\cal P}(F)$,
\be
\label{cond4}
P_{n}(\mF\in B|\;|\mF|=k) = \frac{\sum_{G\in B}\alpha_n}{{2^n\choose k}\alpha_n}=\frac{|B|}{{2^n\choose k}}.
\ee
Denote by 
\[
F^{(k)} = \{G\in\mP(F): |G|=k\}, 
\]
the collection  of binary-function classes of cardinality $k$,
$1 \leq k\leq 2^n$.
Consider the  uniform probability distribution on  $F^{(k)}$
which is defined as follows:
given parameters $n$ and $1\leq k\leq 2^n$ then for any $G\in{\cal P}(F)$,
\begin{equation}
\label{uniformclass}
P^*_{n,k}(G) =\frac{1}{\binom{2^n}{k}}, \mbox{ if } G\in{F}^{(k)}, 
\end{equation}
and $P^*_{n,k}(G)=0$ otherwise.
Hence from (\ref{cond4}) and (\ref{uniformclass}) it follows that for any $B\subseteq\mP(F)$,
\be
\label{equiv}
P_n(\mF\in B|\;|\mF|=k) = P^*_{n,k}(\mF\in B).
\ee
It will be convenient to use  another probability distribution which estimates $P^*_{n,k}$ and is defined
as follows.
Construct a random $n\times k$ binary matrix 
by  fair-coin tossings with the $nk$ elements taking 
values  $0$ or $1$ independently with probability $1/2$.
Denoting  by $Q^*_{n,k}$ the probability measure 
corresponding to this process then
for any matrix  $U\in{\cal U}_{n\times
k}(\{0,1\})$, 
\[
Q^*_{n,k}(U) = \frac{1}{2^{nk}}.
\]
Clearly, the columns of a binary matrix $U$ are vectors of length $n$ which
are  binary functions on $[n]$. Hence the set $\text{cols}(U)$ of columns of 
$U$  represents a class of binary functions. It contains
$k$ elements if and only if  $\text{cols}(U)$
consists of distinct elements,
or less than $k$ elements if two columns are the same. Denote
 by $S$ a {\em simple} binary matrix  as one all of whose  columns are distinct (\cite{AnsteeFlemingFurediSali05}).
 We claim that
the conditional distribution of the set
 of columns of a random binary
matrix,  knowing that the matrix is simple, is the uniform
probability distribution $P^*_{n,k}$. To see this,  observe that the
probability that the columns of a random binary matrix are
distinct is
\begin{equation}
\label{probaS}
Q^*_{n,k}(S) = \frac{2^n(2^n-1)\cdots(2^n-k+1)}{2^{nk}}.
\end{equation}
For any fixed class $G\in\mP(F)$ of $k$ binary functions 
 there are $k!$ corresponding simple matrices in ${\cal U}_{n\times k}(\{0,1\})$.
 Therefore given a simple matrix $S$, the probability that  $\text{cols}(S)$ 
 equals a class $G$ 
is
\begin{equation}
\label{Ps}
Q^*_{n,k}(G|S) = 
\frac{k!}{2^{nk}}\frac{2^{nk}}{2^n(2^n-1)\cdots(2^n-k+1)}
=\frac{1}{\binom{2^n}{k}} = P^*_{n,k}(G).
\end{equation}
Using the distribution
$Q^*_{n,k}$ enables
simpler computations of  the asymptotic probability of  several
types of events that are associated with the properties of Theorems \ref{p1} -- \ref{cor2}.
We  henceforth resort to the following process  for generating a random class $G$:
for every $1\leq k\leq 2^n$ we
repeatedly and  independently draw  matrices of size $n\times k$ using $Q^*_{n,k}$
 until we get a simple matrix $M_{n\times k}$.
 Then we randomly draw a value for $k$ according to the distribution of (\ref{pfk})
and choose the formerly generated simple matrix corresponding to this chosen $k$.
 Since this is a simple matrix then by (\ref{Ps}) it is  clear that  this choice
yields
a random class $G$
which is distributed uniformly in $F^{(k)}$ according to $P^*_{n,k}$.
This is stated formally in Lemma \ref{lem1} below but first we have an auxiliary lemma
that shows the above process  converges.
\begin{lemma}
\label{lemQS}
Let $n=1,2,\ldots$ and consider the process of 
 drawing sequences $s^{(k)}_m = \{M^{(i)}_{k,n}\}_{i=1}^m$,  $1\leq k \leq 2^n$,
all of length $m$
 where the $k^{th}$ sequence consists of  matrices $M^{(i)}_{k,n}\in \mathcal{U}_{n\times k}(\{0,1\})$ 
 which are randomly and independently drawn 
according to the probability distribution $Q^*_{n,k}$. 
Then the probability that after $m=ne^{2^{n}}$ trials there exists a $k$
 such that no simple matrix
appears in  $s^{(k)}_m$, converges to zero  with increasing $n$.
\end{lemma}
\begin{proof}
Let 
\(
a({n,k}) = Q^*_{n,k}(S)
\)
be the probability of getting a simple matrix $M_{n,k}\in\mathcal{U}_{n\times k}$.
Then the probability that there exists some $k$ such that $s^{(k)}_m$ consists of only non-simple matrices
is 
\be
\label{rightsum}
q(n,m) \leq \sum_{k=1}^{2^n}(1-a({n,k}))^m \leq \sum_{k=1}^{2^n} e^{-ma(n,k)}.
\ee
From (\ref{probaS})
 we have
 \be
 \ln a(n,k)  = \sum_{i=0}^{k-1}\ln(2^n - i) -nk\ln 2 = \sum_{j=2^n - (k-1)}^{2^n}\ln j -nk\ln 2.
\ee
Since $\ln x$ is increasing function of $x$ then for any pair of positive integers $2\leq a \leq b$ we have
\[
\sum_{j=a}^b \ln j \geq \int_{a-1}^b \ln x dx = b(\ln b -1) - (a-1)(\ln(a-1) - 1).
\]
Hence 
\[
\ln a(n,k) \geq 2^n(n\ln 2 - 1) - (2^n - k)(\ln(2^n - k) - 1) - nk\ln 2 \equiv b(n,k)
\]
and the right side of (\ref{rightsum}) is now bounded as follows
\be
\label{abb}
\sum_{k=1}^{2^n}  e^{-ma(n,k)}    \leq \sum_{k=1}^{2^n}e^{-m e^{b(n,k)}}.
\ee
From a simple check of the  derivative of $b(n,k)$ with respect to $k$
it follows that  $b(n,k)$  is a decreasing function of $k$ on  $1\leq k\leq 2^n$.
Replacing each term in the sum on the right side of (\ref{abb}) by the last term gives
the following bound
\be
q(n,m) \leq 
 e^{n\ln 2 - m  e^{-2^n} }.
\ee
The exponent is negative  provided
\[
m > n\ln(2) e^{2^n}.
\]
Choosing $m=ne^{2^n}$ guarantees that $q(n,m) \rightarrow 0$ with increasing $n$.
\end{proof}

The  following result  states that the measure $Q^*_{n,k}$ may replace $P^*_{n,k}$ uniformly over $1\leq k\leq 2^n$.
\begin{lemma}
\label{lem1}
Let $B\subseteq{\cal P}(F)$. Then 
\[
\max_{1\leq k\leq 2^n}|P^*_{n,k}(B) - Q^*_{n,k}(B)| \rightarrow 0
\]
as $n$ tends to infinity.
\end{lemma}
\begin{proof}
From (\ref{Ps}) we have
\[
P^*_{n,k}(B)=
Q^*_{n,k}(B|S)=
\frac{Q^*_{n,k}(B\bigcap S)}
{Q^*_{n,k}(S)}.
\]
Then
\bqq
\lefteqn{\max_{k}|P^*_{n,k}(B) - Q^*_{n,k}(B)| =
 \max_{k}\left|\frac{Q^*_{n,k}(B\bigcap S)} {Q^*_{n,k}(S)}  - Q^*_{n,k}(B)  \right|} \nonumber\\
& \leq & \max_{k}\left|\frac{1}{Q^*_{n,k}(S)}\right| \max_{k} \left| Q^*_{n,k}(B\bigcap S) - Q^*_{n,k}(B) Q^*_{n,k}(S) \right|  \nonumber\\
&\leq & \max_{k}\left|\frac{1}{Q^*_{n,k}(S)}\right| 
 \Biggl(
     \max_{k}\left|Q^*_{n,k}(B\bigcap S) - Q^*_{n,k}(B)  \right| + 
     \max_k \left| Q^*_{n,k}(B) (1-Q^*_{n,k}(S))  \right|
 \Biggr) \nonumber\\
\label{coldplay}
&\leq & \max_{k}\left|\frac{1}{Q^*_{n,k}(S)}\right| 
 \Biggl(
     \max_{k}\left|Q^*_{n,k}(B\bigcap S) - Q^*_{n,k}(B)  \right| + 
     \max_k \left| Q^*_{n,k}(B)\right|\max_k \left|1-Q^*_{n,k}(S)\right|
 \Biggr).\quad\quad
\eqq
From Lemma \ref{lemQS} it follows that
\be
\label{fol}
\max_k \left|1-Q^*_{n,k}(S)\right|\rightarrow 0, \qquad
\max_k|1/Q^*_{n,k}(S)|\rightarrow 1
\ee
 with increasing $n$.
For any $1\leq k\leq 2^n$,
\[
Q^*_{n,k}(B)+Q^*_{n,k}(S)-1
\leq Q^*_{n,k}(B\bigcap S)
\leq
Q^*_{n,k}(B)
\]
and
by Lemma \ref{lemQS}, $Q^*_{n,k}(S)$ tends to $1$ uniformly over $1\leq k\leq 2^n$ with increasing $n$. Hence
$\max_{k}\left|Q^*_{n,k}(B\bigcap S) - Q^*_{n,k}(B)  \right| \rightarrow 0$ which together with (\ref{coldplay}) and (\ref{fol})
implies the statement of the lemma.
\end{proof}
We now proceed to the proofs of the theorems in  Section \ref{sec4}.
\subsection{Proofs}
\label{proofs}

Note that for any property $\mM$, the quantity  $\omega_{x}(k)$ in (\ref{newIxy}) is the ratio of the number of classes $G\in F^{(k)}$ that satisfy $\mM$ to the total number
of classes that satisfy $\mM$. It is therefore equal to $P_n(|\mF|=k\,|\,\mF\models\mM)$.
Our approach starts by computing the probability $P_{n}(\mF\models\mM\,|\, |\mF|=k)$ from which 
 $P_n(|\mF|=k\,|\,\mF\models\mM)$ and then $\omega_{x}(k)$ are obtained.

\subsubsection{Proof of Theorem \ref{p1}}
\label{prfp1}
We start with an  auxiliary  lemma which states that  the probability $P_{n}(\mF\models\mL_d\,|\, |\mF|=k)$ possesses a zero-one behavior.
\begin{lemma}
\label{lem2}
Let $\mF$ be a class of cardinality $k_n$ and randomly drawn according to the uniform probability distribution $P^*_{n,k_n}$ on $F^{(k_n)}$.
Then as $n$ increases,  the probability $P^*_{n,k_n}(\mF\models\mL_d)$ that $\mF$
 satisfies property $\mL_d$  tends to $0$ or $1$ if
  $k_n\gg\log(2n/d)$ or
 $k_n = 1+\kappa_n$, $\kappa_n\ll(\log(n))/d$, respectively.
\end{lemma}
\begin{proof}
For brevity, we sometimes write  $k$ for $k_n$.
Using Lemma \ref{lem1} it suffices to show that $Q^*_{n,k}(\mF\models\mL_d)$ tends to $1$ or $0$ under the stated conditions.
For any set $S\subseteq[n]$, $|S|=d$ and any fixed  $v\in\{0,1\}^d$, under the probability distribution $Q^*_{n,k}$,  the event $E_v$ that every function $f\in\mF$ satisfies
$f_{|S}=v$ has a probability $(1/2)^{kd}$.
Denote by $E_S$ the event that all functions in the random class $\mF$ have the same restriction on $S$. 
There are $2^d$ possible restrictions on $S$ and the events $E_v$, $v\in\{0,1\}^d$, are disjoint. Hence  $Q^*_{n,k}(E_S) = 2^d(1/2)^{kd}=2^{-(k-1)d}$.
The  event that $\mF$ has property $\mL_d$, i.e., that $L(\mF)\geq d$, equals the union of $E_S$, over all $S\subseteq[n]$ of cardinality $d$.
Thus we have
\begin{eqnarray*}
Q^*_{n,k}(\mF\models\mL_d)&=&Q^*_{n,k}\left(\bigcup_{S\subseteq[n]:|S|=d}E_{S}\right)\\
&\leq& {n\choose d}Q^*_{n,k}\left(E_{[d]}\right) = 2^{-(k-1)d}\frac{n^d(1-o(1))}{d!}.
\end{eqnarray*}
For $k=k_n\gg\log(2n/d)$ the right  side tends to zero which proves the first statement.
Let the mutually disjoint sets
$S_i=\{id+1,id+2,\ldots,d(i+1)\}\subseteq[n]$, $0\leq i\leq m-1$ where $m=\lfloor n/d\rfloor$. 
The event that $\mM_d$ is not true equals $\bigcap_{S:|S|=d}\overline{E}_S$. Its probability is
\bq
Q^*_{n,k}\left(\bigcap_{S:|S|=d}\overline{E}_S\right) & =& 1-Q^*_{n,k}\left(\bigcup_{S:|S|=d}{E}_S\right)
\leq  \max\{0, 1 - Q^*_{n,k}\left(\bigcup_{i=0}^{m-1}{E}_{S_i}\right)\}.
\eq
Since the sets are disjoint and of the same size $d$ then the  right hand side equals $\max\{0,1-mQ^*_{n,k}({E}_{[d]})\}$. This equals
\[
\max\{0, 1-\left\lfloor\frac{n}{d}\right\rfloor 2^{-(k-1)d}\}
\]
which tends to zero when $k=k_n = 1 + \kappa_n$, $\kappa_n\ll(\log(n))/d$. The second statement is proved.\end{proof}
%
\begin{remark}
While from this result it is clear that the critical value of  $k$ for the conditional probability $P_n(\mL_d|k)$ to tend to $1$
is $O(\log (n))$, as will be shown below,
when considering the conditional probability $P_n(k|\mL_d)$,  the most probable value of $k$ 
 is much higher at $O(2^{n-d})$.
\end{remark}

We continue now with  the  proof of Theorem \ref{p1}.
For any probability measure $\prob$ on $\mP(F)$  denote by 
$
\prob(k|\mL_d)=\prob(|\mF|=k\,|\,\mF\models\mL_d).
$
By the premise of Theorem \ref{p1}, the input $x$ describes  the target $t$ as an element of a class that satisfies property $\mL_d$.
  In this case the quantity $\omega_{x}(k)$ is the ratio of the number of classes of cardinality $k$ that satisfy $\mL_d$ to the total number
  of classes that satisfy $\mL_d$.
Since by (\ref{unif1}) the probability distribution $P_n$ is uniform over the space $\mP(F)$ whose size is $2^{2^n}$ then
\be
\label{omega2}
\omega_{x}(k) = \frac{P_n(k, \mL_d) 2^{2^n} }{P_n(\mL_d) 2^{2^n}}  = P_n(k|\mL_d).
\ee
We have
\[
P_{n}(k|\mL_d) =\frac{P_{n}(\mL_d|k)P_{n}(k)}{\sum_{j=1}^{2^n}P_{n}(\mL_d|j)P_{n}(j)}.
\]
By  (\ref{equiv}), it follows therefore that 
the sum  in (\ref{newIxy})  equals
\be
\label{ratio}
\sum_{k=2}^{2^n}\omega_{x}(k)\log(k) = 
\sum_{k=2}^{2^n}
\frac{P^*_{n,k}(\mL_d) P_n(k)}{\sum_{j=1}^{2^n}P^*_{n,j}(\mL_d) P_n(j)}\log k.
 \ee
Let $N=2^n$, then
by Lemma \ref{lem1} and from the proof of Lemma \ref{lem2}, as $n$ (hence $N$) increases, it follows that
\be
\label{pnk1}
P^*_{n,k}(\mL_d) \approx Q^*_{n,k}(\mL_d) = \left(\frac{1}{2}\right)^{d(k-1)} A(N,d)
\ee
where $A(N,d)$ satisfies
\[
 \frac{\log N}{d}\leq A(N,d) \leq \frac{\log^dN}{d!}.
\]
Let $p=1/(1+2^{d})$ then using (\ref{pnk1})
the ratio in (\ref{ratio}) is 
\be
\label{ap34}
\frac{\sum_{k=2}^N {N\choose k}p^k(1-p)^{N-k}\log k}{\sum_{j=1}^N {N\choose j}p^j(1-p)^{N-j}}.
\ee
Substituting for $N$ and $p$, the denominator equals
\be
\label{denum3}
1-(1-p)^N  =1 - \left(1 - \frac{1}{1+2^d}\right)^{2^n} = 1 - \left(\frac{2^d}{1+2^d}\right)^{2^n}.
\ee 
Using the DeMoivre-Laplace limit theorem \cite{Feller1}, the binomial distribution 
$P_{N,p}(k)$ with parameters $N$ and $p$ satisfies
\[
P_{N,p}(k)\approx\frac{1}{\sigma}\phi\left(\frac{k-\mu}{\sigma}\right), \;\; N\rightarrow\infty
\]
where
$\phi(x) = (1/\sqrt{2\pi})\exp(-x^2/2)$ is the standard normal probability density function and
 $\mu=Np$,  $\sigma = \sqrt{Np(1-p)}$.
The sum in the numerator of (\ref{ap34}) may be approximated by an integral
\[
\frac{1}{\sigma}\int_2^\infty \phi\left(\frac{x-\mu}{\sigma}\right)\log x\, dx
 = \int_{(2-\mu)/\sigma}^\infty \phi(x)\log (\sigma x +\mu) dx.
\]
The $\log$ factor equals  $\log \mu$ + $\log(1+ x\sigma/\mu) =\log \mu + x\sigma/\mu + O(x^2(\sigma/\mu)^2)$.
Denote by 
\(
a=(2-\mu)/\sigma
\)
 then the right  side above equals
\[
\Phi(-a)\log \mu + \frac{\sigma}{\mu}\phi(a) + O\left(\left(\frac{\sigma}{\mu}\right)^2\right)
\]
where $\Phi(x)$ is the normal cumulative probability  distribution.
Substituting for $\mu$, $\sigma$,  $p$ and $N$, and combining with 
 (\ref{denum3}) then (\ref{ratio}) is asymptotically equal to 
\be
\label{phia}
\sum_{k=2}^{2^n}\omega_{x}(k)\log k \approx
\frac{
 \Phi\left(-a\right)\log\left(\frac{2^n}{1+2^{d}}\right) + 
2^{-(n-d)/2}\phi(a) + O(2^{-(n-d)})}
{1 - \left(\frac{2^d}{1+2^d}\right)^{2^n}}
\ee
where
\[
a =  2(1+2^{d})2^{-(n+d)/2} - 2^{(n-d)/2}.
\]
In Theorem \ref{p1}, the set $\bY$    is the class $F$ (see (\ref{newIxy})) hence $\log|F| = n$ and
\bq
I(x:F) &=& n -  \sum_{k\geq 2}\omega_{x}(k)\log k.
\eq
Combining with  (\ref{phia})
the first statement of Theorem \ref{p1} follows.

We now compute the description complexity $\ell(x_{\mL_d})$.
Since  in this setting $\bY$ is $F$ and $\bZ$ is $\mP(F)$ then,
by (\ref{lx1}), 
the description complexity $\ell(x_{\mL_d})$ is $2^n - \log|\mL_d|$.
Since , the probability distribution $P_n$ is uniform on $\mP(F)$
then the cardinality of $\mL_d$ equals
\[
|\mL_d| = 2^{2^n}P_n(\mL_d).
\]
It follows that 
\[
\ell(x_{\mL_d}) = -\log P_n(\mL_d)
\]
hence it suffices to compute  $P_n(\mL_d)$.
Letting $N=2^n$, we have 
\[
P_n(\mL_d) = \sum_{k=1}^NP_n(\mL_d|k)P_N(k) = \sum_{k=1}^NP^*_{n,k}(\mL_d)P_N(k).
\]
Using (\ref{pnk1}) and
letting $p=1/(1+2^d)$ this becomes
\bq
\lefteqn{\sum_{k=1}^N\left(\frac{1}{2}\right)^{d(k-1)}A(N,d)\left(\frac{1}{2}\right)^{N}{N\choose k}}\\
&=&
\left(\frac{1}{1-p}\right)^N\left(\frac{1}{2}\right)^{N-d}A(N,d)(1-(1-p)^N).
\eq
Letting $q = (1-p)^N$, it follows that
\bq
-\log P_n(\mL_d) &=& \log\left(\frac{2^{N-d}}{A(N,d)}\right) + \log\left(\frac{q}{1-q}\right)\\
& = & N-d -\log A(N,d) + q + O(q^2)+ \log(q)\\
& =&  N(1-p-O(p^2))-d - c\log\log N + o(1)
\eq
where $1\leq c \leq d$.
Substituting for $N$ gives the result.
 \hfill\bqed

\subsubsection{Proof of Theorem \ref{p2}}
We start with an auxiliary lemma 
 that states
a  threshold value for the cardinality of a random element of $F^{(k)}$ that satisfies property ${\mV}^c_d$.
\begin{lemma}
\label{lem3}
For any integer $d>0$ let $k$ be an integer satisfying $k\geq 2^d$. 
Let $\mF$ be a class of cardinality $k$ and randomly drawn according to the uniform probability distribution $P^*_{n,k}$ on $F^{(k)}$.
Then
\[
\lim_{n\rightarrow\infty} P^*_{n,k}( {\mV}^c_d ) =1.
\]
\end{lemma}
\begin{remark}
\label{threshEd}
When $k_n< 2^d$ there does not exist an $E\subseteq[n]$ with $|\tr_E(\mF)|=2^d$ 
hence  $P^*_{n,k_n}({\mV}^c_d) = 0$. For $k_n\gg 2^d$,
$P^*_{n,k_n}({\mV}^c_d) $ tends to $1$.  Hence, for a random class $\mF$
 to have property ${\mV}^c_d$
 the critical value 
of its cardinality is $2^d$.
\end{remark}
We proceed now with the proof of  Lemma \ref{lem3}.
\begin{proof}
It suffices to prove the result for $k=2^d$ since
$P^*_{n,k}(\mF\models{\mV}^c_d)\geq P^*_{n,2^d}(\mF\models{\mV}^c_d)$. 
As in the proof of Lemma \ref{lem2}, we  represent $P^*_{n,2^d}$ by $Q^*_{n,2^d}$ using (\ref{Ps}) and
with 
 Lemma
\ref{lem1}  it suffices to show that $Q^*_{n,2^d}(\mF\models{\mV}^c_d)$
tends to $1$. Denote by $U_d$ the `complete' matrix with $d$ rows and $2^d$
columns formed by all $2^d$ binary vectors of
length $d$, ranked for instance in alphabetical order.
The event ``$\mF\models{\mV}^c_d$'' occurs if there exists
a subset $S=\{i_1,\ldots,i_d\}\subseteq [n]$ such that the submatrix
whose rows are indexed by $S$ and columns by $[2^d]$, is
equal to $U_d$. 
 Let
$S_i=\{id+1,id+2,\ldots,d(i+1)\}$, $0\leq i\leq m-1$, be the sets defined in the proof of Lemma \ref{lem2}
and consider the $m$ corresponding events which are defined as follows:
the $i^{th}$ event  is  described as
having a submatrix whose rows are indexed by $S_i$ 
and is equal to $U_d$.
Since the sets $S_i$, $1\leq i\leq m$ are disjoint it is clear that these events are
independent and have the same probability 
\[
Q^*_{n,k}(S_i)=2^{-d2^d}.
\] 
Hence the probability that at least one of them is
fulfilled is
\[
1-(1-2^{-d2^d})^{\lfloor n/d\rfloor}
\]
which tends to $1$ as $n$ increases.
\end{proof}
We continue with the proof of Theorem \ref{p2}.
As in the proof of Theorem \ref{p1}, 
since by (\ref{unif1})
 the probability distribution $P_n$ is uniform over $\mP(F)$  then 
\[
\omega_{x}(k) = \frac{P_n(k, {\mV}^c_d) 2^{2^n} }{P_n({\mV}^c_d) 2^{2^n}}  = P_n(k|{\mV}^c_d) =  \frac{P^*_{n,k}({\mV}^c_d)P_n(k)}{\sum_{j=0}^{2^n}P^*_{n,j}({\mV}^c_d)P_n(j)}, 1\leq k\leq 2^n.
\]
Considering  Remark \ref{threshEd}, in this case the sum in (\ref{newIxy}) is 
\be
\label{s14}
 \sum_{k=2^d}^{2^n}\frac{P^*_{n,k}({\mV}^c_d)P_n(k)\log k}{\sum_{j=2^d}^{2^n}P^*_{n,j}({\mV}^c_d)P_n(j)}.
\ee
We now obtain its asymptotic value as $n$ increases.
From the proof of  Lemma \ref{lem3}, 
it  follows that
 for all $k\geq 2^d$, 
 \[
 P^*_{n,k}({\mV}^c_d) \approx 1-(1-\beta)^{rk}, \;\beta = 2^{-d2^d}, r={\frac{n}{d2^d}}.
  \]
Since $\beta$ is an exponentially small positive real  we approximate $(1-\beta)^{rk}$ by $1-rk\beta$
(by assumption, $n< d2^{d}$ hence this remains positive for all $1\leq k\leq 2^n$).
Therefore we  take
 \be
 \label{bet}
 P^*_{n,k}({\mV}^c_d)\approx rk\beta
 \ee
 and (\ref{s14}) is approximated by
 \be
 \label{bphi}
\frac{\sum_{k=2^d}^{2^n}kP_n(k)\log k}{\sum_{j=2^d}^{2^n}jP_n(j)}.
 \ee
 As before, for simpler notation let us denote   $N=2^n$ and  let  $P_N(k)$ be the binomial distribution with parameters $N$ and $p=1/2$.
Denote by $\mu=N/2$ and $\sigma=\sqrt{N/4}$,
  then
 using the DeMoivre-Laplace limit theorem 
 we have 
 \[
 P_N(k) \approx \frac{1}{\sigma}\phi\left(\frac{k-\mu}{\sigma}\right), \; N\rightarrow\infty.
\]
Thus (\ref{bphi}) is approximated by the ratio of two integrals
\be
\label{alm}
\frac{\int_{(2^d-\mu)/\sigma}^\infty\phi(x)(\sigma x +\mu)\log(\sigma x+\mu)dx}{\int_{(2^d-\mu)/\sigma}^\infty\phi(x)(\sigma x +\mu)dx}.
\ee
The $\log$ factor equals  $\log \mu$ + $\log(1+ x\sigma/\mu) =\log \mu + x\sigma/\mu + O(x^2(\sigma/\mu)^2)$.
Denote by 
\be
\label{a}
a=(\mu-2^d)/\sigma
\ee
 then  the numerator is approximated by
\bq
\lefteqn{\Phi(a)\mu\log\mu + \sigma(1+\log\mu)\phi(a) + O\left(\frac{\sigma^3}{\mu^2}a^2\phi(a)\right)}\\
& =& \log(N/2)\left(\Phi(a)N/2 + \left(1+\frac{a^2}{N\log(N/2)}\right)\phi(a) \sqrt{N}/2\right).
\eq
Similarly,  the denominator of (\ref{alm}) is  approximated by
\(
\Phi(a)N/2 + \phi(a) \sqrt{N}/2.
\)
The ratio, and hence  (\ref{s14}),  tends to 
\[
\frac{\log(N/2)\left(\Phi(a)N/2 + \left(1+\frac{a^2}{N\log(N/2)}\right)\phi(a) \sqrt{N}/2\right)}
{\Phi(a)N/2 + \phi(a) \sqrt{N}/2}.
\]
Substituting back for $a$  then the above tends to
$\log(N/2) = \log N -1$. With
$N=2^n$ and (\ref{newIxy})
the first statement of the theorem follows.

We now compute the description complexity $\ell(x_{{\mV}^c_d})$. 
Following the steps of the second part of the proof of Theorem \ref{p1} (Section \ref{prfp1})
we have $\ell(x_{{\mV}^c_d})=-\log P_n({\mV}^c_d)$.
Using (\ref{bet})
the probability is approximated by
\[
P_n({\mV}^c_d) \approx r\beta\sum_{k=2^d}^{N} kP_N(k)
\]
and as before, this is approximated by $r\beta(\Phi(a)N + \phi(a)\sqrt{N})/2$.
Thus substituting for $r$ and $\beta$ we have 
\[
-\log P_n({\mV}^c_d) \approx  d(2^d +1) + \log(d) - \log\left(\Phi(a)N + \phi(a)\sqrt{N}\right) -\log\log N +1.
\]
Substituting for $N$
 yields the result.
 \hfill\bqed

\subsubsection{Proof of Theorem \ref{cor1}}
%
%

The proof is almost identical to that of Theorem \ref{p2}.
From (\ref{bet}) we have
 \[
  P^*_{n,k}({\mV}_d) = 
  1 - P^*_{n,k}({\mV}^c_d)
     =
     \left \{  \begin{array}{ll}
    1   &,  1\leq k < 2^d\\
    1- P^*_{n,k}(\mV_d^c) \approx 1 - r\beta k   &, 2^d \leq k\leq 2^n
\end{array}
\right.
  \]
hence 
\be
\label{vc==d}
 \sum_{k=2}^{2^n}\frac{P^*_{n,k}({\mV}_d)P_n(k)\log k}{\sum_{j=1}^{2^n}P^*_{n,j}({\mV}_d)P_n(j)} \approx
 \frac{\sum_{k=2}^{2^n}P_n(k)\log k - r\beta\sum_{k=2^d}^{2^n}kP_n(k)\log k}{\sum_{j=1}^{2^n}P_n(j)-r\beta\sum_{j=2^d}^{2^n}jP_n(j)}.
\ee
Let $a$ be as in (\ref{a}) and denote by  $b=(\mu-2)/\sigma$ and 
\[
s = \Phi(a)N/2 + \phi(a) \sqrt{N}/2
\]
then the numerator tends to
\(
\log(N/2) (\Phi(b) - rs\beta) + \phi(b)/\sqrt{N}
\)
and the denominator tends to $1-(1/2)^N - rs\beta$.
Then (\ref{vc==d}) tends to
\[
\log (N/2) \frac{\Phi(b)-rs\beta}{1-rs\beta} + \frac{\phi(b)}
{\sqrt{N}\left(1-(1/2)^N - rs\beta  \right)} \approx \log(N/2) + \frac{\phi(b)}{\sqrt{N}(1-rs\beta)}.
\]
Substituting for $r$, $\beta$ and $N$ yields  the statement of the theorem.\bqed

\subsubsection{Proof of Theorem \ref{cor2}}

The probability that a random class of cardinality $k$ satisfies the property $\mV_d(S_m)$ is
\begin{eqnarray}
\lefteqn{P_n(\mF\models\mV_d(S_m)\;|\; |\mF|=k) = P_n(\mF\models\mV_d, \mF_{|\xi}=\zeta \;|\; |\mF|=k)\nonumber} \quad\quad\\
\label{fgf}
&=& P_n(\mF\models\mV_d \;|\; \mF_{|\xi}=\zeta, |\mF|=k) P_n(\mF_{|\xi}=\zeta\;|\; |\mF|=k).
\end{eqnarray}
The  factor on the right of (\ref{fgf}) is the probability of the condition $E_\zeta$ that a random class $\mF$ of size $k$
 has for all its elements the same restriction  $\zeta$
on the sample $\xi$. As in the proof of Lemma \ref{lem2} it suffices
to use the probability distribution $Q^*_{n,k}$
in which case $Q^*_{n,k}(E_\zeta)$ is $\gamma^k$ where $\gamma\equiv(1/2)^{m}$.
The left factor of (\ref{fgf}) is the probability
that a random class $\mF$ with cardinality $k$ which satisfies $E_\zeta$
will satisfy property $\mV_d$.
This is the same as the event that a random class $\mF$ on $[n]\setminus S_m$
satisfies property $\mV_d$. Its probability is $P^*_{n-m,k}(\mV_d)$ 
which equals $1$ for $k<2^d$ and
using (\ref{bet}) for $k\geq 2^d$ it  is approximated as
\[
1 - P^*_{n-m,k}({\mV}^c_d)    \approx 1 - rk\beta
\]
where $r =(n-m)/(d2^d)$ and $\beta = 2^{-d2^d}$. 
Hence the conditional entropy becomes
\be
\label{vc=d}
 \sum_{k=2}^{2^n}\frac{\gamma^kP^*_{n-m,k}({\mV}_d)P_n(k)\log k}
 {\sum_{j=1}^{2^n}\gamma^jP^*_{n-m,j}({\mV}_d)P_n(j)} 
 \approx
 \frac{\sum_{k=2}^{2^n}\gamma^{k}P_n(k)\log k - r\beta\sum_{k=2^d}^{2^n}\gamma^{k}kP_n(k)\log k}
 {\sum_{j=1}^{2^n}\gamma^{j}P_n(j)-r\beta\sum_{j=2^d}^{2^n}\gamma^{j}jP_n(j)}.
\ee
Let $p=\gamma/(1+\gamma)$, $N=2^n$ and denote by $P_{N,p}(k)$ the binomial distribution with parameters $N$ and $p$.
Then (\ref{vc=d}) becomes
 \be
 \label{vc--d}
 \frac{\sum_{k=2}^{N}P_{N,p}(k)\log k - r\beta\sum_{k=2^d}^{N}kP_{N,p}(k)\log k}
 {\sum_{j=1}^{N}P_{N,p}(j)-r\beta\sum_{j=2^d}^{N}jP_{N,p}(j)}.
 \ee
With $\mu=Np$ and $\sigma=\sqrt{Np(1-p)}$
let $a=(\mu-2^d)/\sigma$,  $b=(\mu-2)/\sigma$ and
\[
s = \Phi(a)Np + \phi(a) \sqrt{Np(1-p)}
\]
then the numerator tends to
\[
\log(Np) (\Phi(b) - rs\beta) + \frac{\phi(b)}{\sqrt{N}}\sqrt{\frac{1-p}{p}}
\]
and the denominator tends to $1-(2^m/(2^m+1))^N - rs\beta$.
%
 Therefore (\ref{vc--d}) tends to
\[
\log (Np) + \frac{\phi(b)}
{\sqrt{N}\left(1-o(1)- rs\beta  \right)}\sqrt{\frac{1-p}{p}}.
 \]
Substituting for $r$, $s$, $\beta$ and $N$ yields  the
first  statement of the theorem.

Next, we obtain the description complexity.
We have
\[
\ell(x_{\mV_d(S_m)}) = -\log P_n(\mV_d(S_m)).
\]
The probability $P_n(\mV_d(S_m))$ is the denominator
of  (\ref{vc=d}) which equals the denominator of (\ref{vc--d}) multiplied by a factor of $(2(1-p))^{-N}$  hence from above
\[
-\log P_n(\mV_d(S_m)) \approx -\log\left(1- \left(\frac{2^m}{1+2^m}\right)^N - r\beta s\right) + N 
+ N\log \left(1-p\right).
\]
Let $q=\left(\frac{2^m}{1+2^m}\right)^N + r\beta s$ then we have as an estimate
\[
\ell(x_{\mV_d(S_m)}) \approx \log \left(\frac{1}{1-q}\right) +N(1+\log(1-p))
\]
from which the second statement of the theorem follows. \bqed

\section{Acknowledgements}

The author thanks Dr. Vitaly Maiorov from the department of mathematics of the Technion
for useful remarks.

\bibliographystyle{plain}  

\appendix\section*{Figures}

\begin{figure}[hbt]
\begin{center}
$\begin{array}{c@{\hspace{.1in}}c}
\multicolumn{1}{l}{} &	\multicolumn{1}{l}{} \\ 
\epsfxsize=2.4in    \includegraphics[clip=true,scale =1,angle=90]{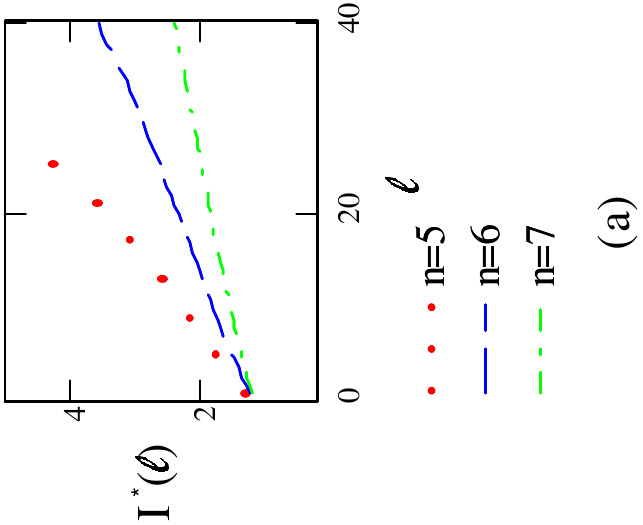} 
& 
{\includegraphics[clip=true,scale = 1,angle=90]{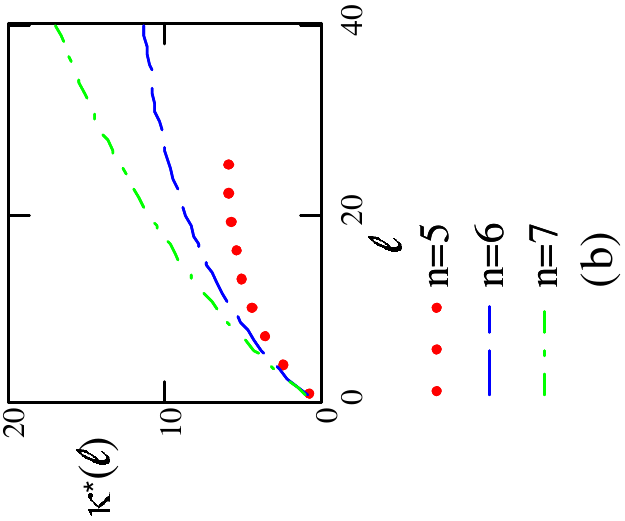} }
\end{array}$
 \end{center}
\caption{ (a) $I^*(\ell)$, ~~ (b) $\cost^*(\ell)$
 }
\label{Istara}
\end{figure}

\begin{figure}[hbt]
\begin{center}

\epsfxsize=3in
\includegraphics[clip=true,scale = 1]{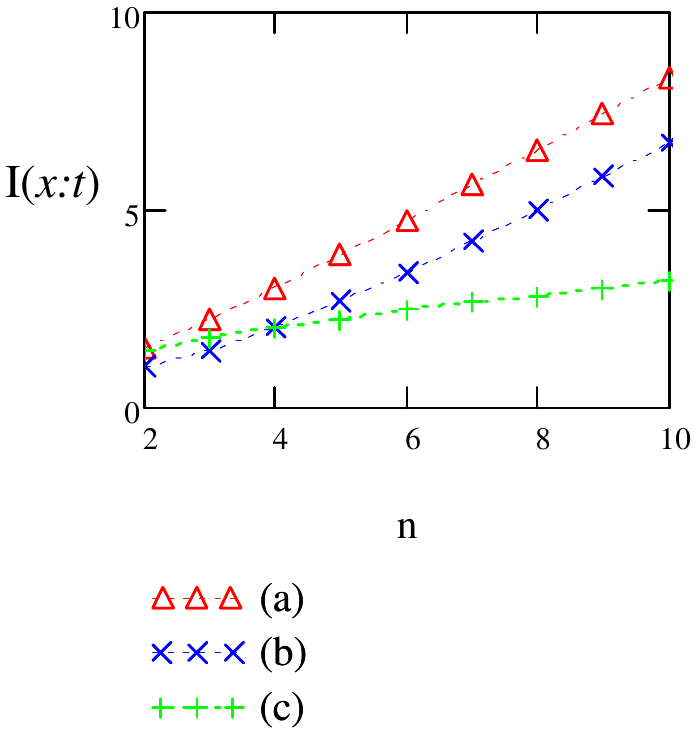}
\end{center}
\caption{Information $I(x_{\mM(G)}:F)$ for  (a) $|G|=\sqrt{n}$,(b) $|G|=n$  and (c)  $|G|=2^{n-\sqrt{n}}$  }
\label{fig2a}
\end{figure}


\begin{figure}[hbt]
\begin{center}

\epsfxsize=3in
\includegraphics[clip=true,scale = 1]{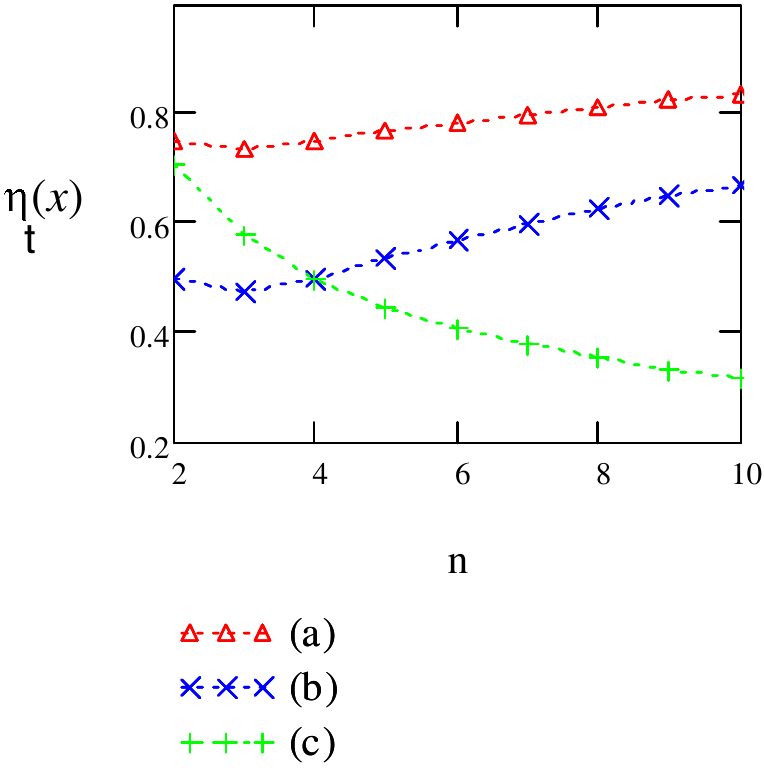}
\end{center}
\caption{Efficiency $\eta_F(x_{\mM(G)})$ for  (a) $|G|=\sqrt{n}$,(b) $|G|=n$  and  (c) $|G|=2^{n-\sqrt{n}}$  }
\label{fig2b}
\end{figure}


\begin{figure}[hbt]
\begin{center}

\epsfxsize=3in
\includegraphics[clip=true,scale = 1]{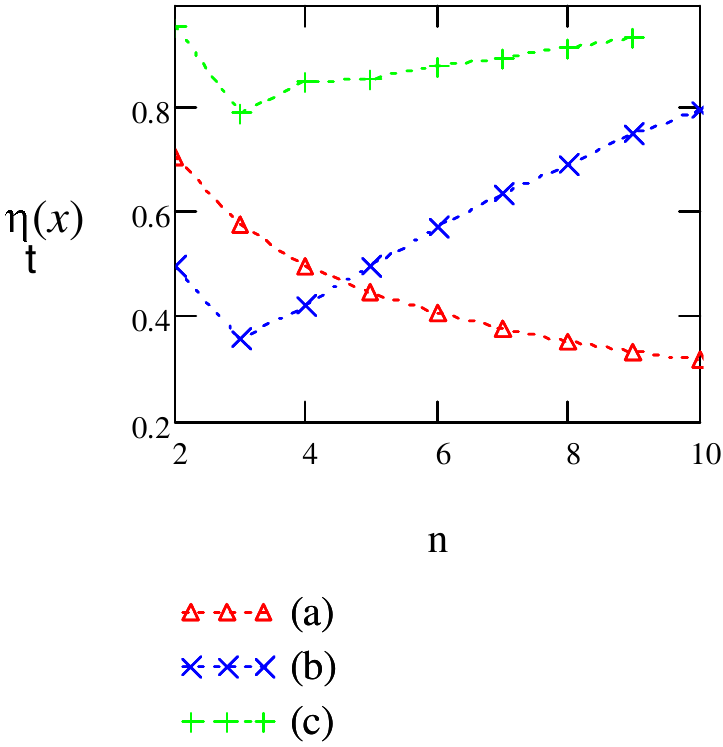}
\end{center}
\caption{Efficiency $\eta_F(x)$ for (a) 
$x_{\mL_d}$, (b) $x_{{\mV}^c_d}$ and (c) $x_{{\mV}_d }$,  $d=\sqrt{n}$ }
\label{fig3a}
\end{figure}

\begin{figure}[hbt]
\begin{center}

\epsfxsize=3in
\includegraphics[clip=true,scale =1]{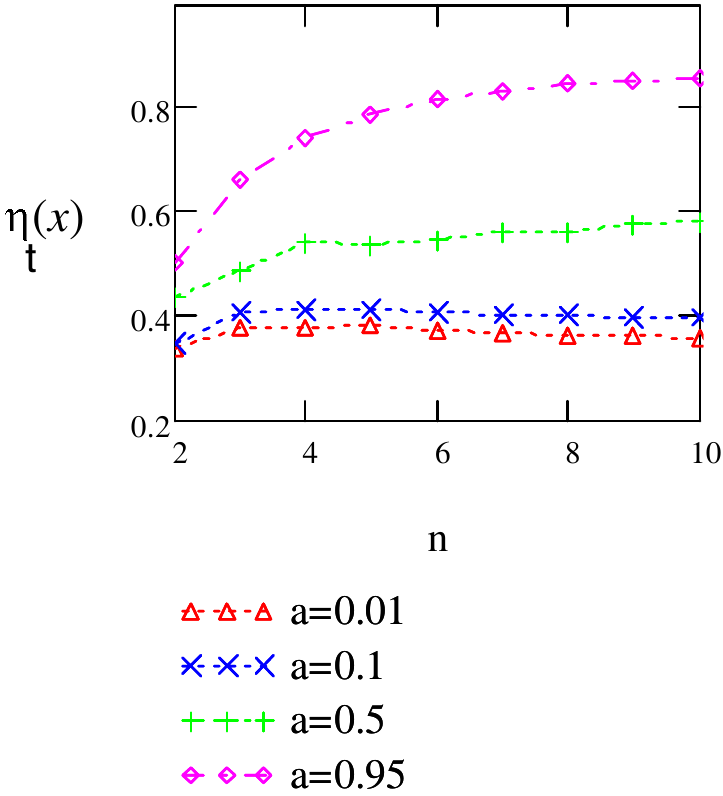} \\ [0.4cm]
\end{center}
\caption{Efficiency $\eta_F(x_{{\mV}_d(S_m)})$,  with $m = n^a$, $a=0.01, 0.1, 0.5, 0.95$, $d=\sqrt{n}$}
\label{fig3b}
\end{figure}

\begin{figure}[hbt]
\begin{center}

\epsfxsize=3in
\includegraphics[clip=true,scale = 1]{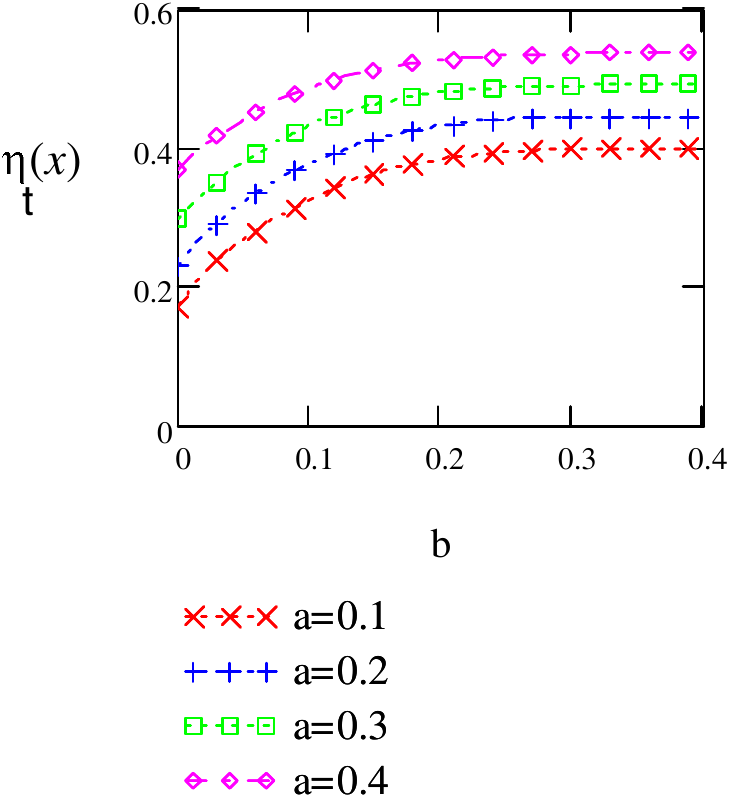}
\end{center}
\caption{Efficiency $\eta_F(x_{{\mV}_d(S_m)})$,  with  $n=10$, $m(n)=m^a$, $d(n)=n^b$}
\label{Fig4}
\end{figure}

\end{document}